\documentclass[aps, prx, superscriptaddress, reprint, floatfix]{revtex4-2}

\usepackage{amsmath, amsthm, amssymb}
\usepackage{mathtools}
\usepackage{xcolor}
\usepackage{braket}
\usepackage{bbold}
\usepackage{bm}
\usepackage{graphicx}
\usepackage{enumitem}
\usepackage{soul} 

\usepackage{tikz}
\usetikzlibrary{decorations.pathmorphing}
\usetikzlibrary{arrows.meta}

\tikzset{
	treenode/.style = {shape=rectangle, rounded corners,
		draw, align=center,
		top color=white, bottom color=blue!20},
	troot/.style     = {treenode, font=\Large, bottom color=red!30},
	tenv/.style      = {treenode, font=\ttfamily\normalsize},
	tleaf/.style      = {treenode, font=\ttfamily\normalsize, bottom color=green!30},
	tdummy/.style    = {circle,draw}
}

\definecolor{myurlcolor}{rgb}{0,0,0.7}
\definecolor{myrefcolor}{rgb}{0.8,0,0}

\usepackage[
	pdftex,
	pdfauthor={Karol Horodecki, Maciej Stankiewicz},
	pdftitle={Semi-Device Independent Quantum Money},
	pdfsubject={Quantum Information Theory},
	pdfkeywords={quantum, quantum money, semi device independent quantum cryptography},
	pdfproducer={Latex with hyperref},
	pdfcreator={pdflatex},
	colorlinks=true, 
	linkcolor=myrefcolor,
	citecolor=myurlcolor,
	urlcolor=myurlcolor
]{hyperref}

\newtheorem{lemma}{Lemma}
\newtheorem{corollary}{Corollary}

\newtheorem{theorem}{Theorem}
\newtheorem*{theorem*}{Theorem}

\newtheorem{observation}{Observation}

\newtheorem{remark}{Remark}
\newtheorem{example}{Example}

\def\>{\rangle}
\def\<{\langle}

\begin{document}

\title{Semi-Device Independent Quantum Money}

\author{Karol Horodecki}
\email[E-mail: ]{khorodec@inf.ug.edu.pl}
\affiliation{Institute of Informatics, National Quantum Information Centre, Faculty of Mathematics, Physics and Informatics, University of Gda{\'n}sk, Wita Stwosza 57, 80-308 Gda{\'n}sk, Poland}
 
\author{Maciej Stankiewicz}
\email[E-mail: ]{maciej@stankiewicz.edu.pl}
\affiliation{Institute of Mathematics, National Quantum Information Centre, Faculty of Mathematics, Physics and Informatics, University of Gda{\'n}sk, Wita Stwosza 57, 80-308 Gda{\'n}sk, Poland}

\date{March 19, 2019}
	
\begin{abstract}
The seminal idea of quantum money not forgeable due to laws of Quantum Mechanics proposed by Stephen Wiesner, has laid foundations for the Quantum Information Theory in early '70s. Recently, several other schemes for quantum currencies have been proposed, all however relying on the assumption that the mint does not cooperate with the counterfeiter. Drawing inspirations from the semi-device independent quantum key distribution protocol, we introduce the first scheme of quantum money with this assumption partially relaxed, along with the proof of its unforgeability. Significance of this protocol is supported by an impossibility result, which we prove, stating that there is no both fully device independent and secure money scheme. Finally, we formulate a quantum analogue of the Oresme-Copernicus-Gresham's law of economy.
\end{abstract}

\maketitle

\section{Introduction}

Quantum Information Science originates from the seminal idea of the scheme of quantum money due to Stephen Wiesner \cite{Wiesner1983}. According to his brilliant concept, the randomly polarized photons could in principle represent the banknote, while the Bank's secret key represents the random choices of polarizations. During verification the Bank checks and accepts the banknote if the photons appear to be polarized as they have been designed and rejects otherwise. Although it is rather intuitive that due to quantum no-cloning \cite{nocloning1,nocloning2,nocloning3} the banknote can not be forged without disturbing it, this scheme, has been proven to be secure against counterfeiter only recently \cite{Molina2013}. Wiesner's scheme bases strongly on the assumption that measurements of the verification are performed according to the specification. Dmitry Gavinsky \cite{Gavinsky2012} has designed a protocol powerful enough to drop this assumption. However security of the latter relies on the honesty of the provider of the banknotes (a possibly malicious \textit{mint}). 

In this manuscript we initiate the study on security of money schemes against \textit{joined} attack of the \textit{mint} and \textit{counterfeiter} who can collaborate by changing the functionality of the inner workings of the terminal that verifies the banknote. We first observe, that a scheme of money unconditionally secure against the joined attack does not exist.  In spite of this fact, we propose
a money scheme with relaxed assumptions about untrusted source \textit{and} untrusted
measurements, and prove its security against a wide class of attacks.
More precisely we show how to change the verification procedure of  Wiesner's banknote in order to assure its security against a variant of the joined attack - the production and counterfeiting performed in a \textit{qubit-by-qubit} manner.

It is easy to see, that protection of Wiesner's banknotes against the joined attack, demands a sufficient control over the banknote state's dimension. At the same time for the banknote to be protected its state  can not be classical (diagonal in single basis). In both cases we demonstrate straightforward attacks. We then start from an observation that the well known quantum cryptographic scheme, the protocol of Paw\l{}owski and Brunner of semi-device independent quantum key distribution (SDI QKD) \cite{sdi} matches these two cases. It (i) assumes that the traveling quantum data have bounded dimension (in considered case the bound is $2$ i.e., we consider qubits) and (ii) assures, via testing an equivalent of a \textit{dimension witness}, that the data are not classical. 

We first note, that according to the honest implementation of the SDI QKD protocol, the sender-receiver state, is exactly in form of Wiesner's banknote. We further propose that the verification of this banknote should be performed exactly as it is done during verification of the SDI QKD protocol. In the latter the honest measuring device does not check correctness of the correlations in the two polarization bases of the original banknote's state, but in the \textit{rotated} bases \cite{CHSH}  as it is specified by the honest implementation of the SDI QKD protocol.

The security analysis of the SDI money scheme needs to take into account that the receiver (Alice) in the corresponding SDI QKD, is \textit{not trusted}. In that sense, the schemes of money are two-party cryptographic problems. In a single verification of a banknote Alice is asked to give certain answers (guessing bits of the key of one branch of the Bank). Upon successful verification, in order to pass the verification for the second time in the other branch she can copy the correct answers given from the first one. We are able to find \textit{the necessary and sufficient value of the threshold} $\theta$ in corresponding SDI QKD protocol which guarantees protection against forgery in the SDI money scheme. Namely, we prove that the owner of a single banknote can not be accepted in two (or a reasonable, i.e., polynomial in banknote's length number of) branches of the Bank, that employs verification with this threshold $\theta$. It is sufficient and also necessary that the threshold is the one that \textit{in the corresponding SDI QKD protocol would imply more than a half of the maximal key rate.} It is important to note that in the SDI money scheme only the preparation and the verification part of the corresponding SDI QKD is performed, while the privacy amplification and information reconciliation is not done. In particular the number of runs is only big enough in order to collect sufficient data for tomography of the guessing probability.

Due to the fact that we base on the original SDI QKD protocol \cite{sdi}, the SDI money scheme inherits the similar security level, which in our context  we call the \textit{qubit-by-qubit} counterfeiting. Major property of this attack is that the malicious mint, during production of Wiesner's banknote can perform a qubit in a different state than given in the specification, independently in each round. Later a  counterfeiter can again try to copy the banknote by individual copy operations applied on each qubit separately. It is also vital for the scheme that the collaboration of the mint and counterfeiter is restricted in a way that mint does not pass a state entangled with the banknote to the counterfeiter. We prove the security in this scenario for the case when the Bank verifies the banknote. We then present also a bit more relaxed case when the counterfeiter can lie in some way about the classical data generated during verification of the banknote. It is known that the banknote in the original Wiesner scheme needs to be destroyed. In the considered case by the nature of measurement test incompatible with the basis of the banknote the honest verification destroys the banknote.

A number (in fact, more than 20) of various quantum money schemes has been recently proposed \cite{Wiesner1983, Bennett1983, AQCash, Aaronson2009, AMS-MosSte10, LutomirskiAFGKHS10, Gavinsky2012, Pastawski2012, Aaronson2012, Farhi2012, Molina2013,  2016arXiv160401383J, 2017arXiv171102276Z, 2017arXiv170804955I, qulogicoin, 1806.05884, 1809.05925, qptmoney}. We then ask if the Oresme-Copernicus-Gresham (OCG) Law of economy (known also as the Gresham's Law \cite{CarolinaSparavigna2014, GreshamArystofanes, GreshamOresme, GreshamCopernicus, Gresham1, Gresham2}) will be also applicable to the quantum schemes of money. If so, \textit{the Quantum Oresme-Copernicus-Gresham law} would have a form:
\begin{align}
	\text{\textbf{\textit{ Bad quantum money drives out}}} \nonumber \\
	\text{\textbf{\textit{ good quantum one}}}\nonumber
\end{align}
We exemplify this general hypothesis on the base of presented 
scheme: the realizations of SDI money schemes with acceptance level $\theta' \geq \theta $ would drive out SDI money schemes with higher acceptance level $\theta''>\theta'$. This is because the banknotes in the latter schemes could be more robust to noise, hence they could in principle by stored for longer. One can expect that in analogy to the OCG law, individuals would tend to keep rather the banknotes that are more robust to noise banknotes, while spending the less robust once more often.

The manuscript is organized as follows. In Section \ref{sec:previous} we review previous quantum money schemes both in private key and public key settings.
In Section \ref{sec:sdi} we present the main results of that work, stating a scheme for a semi-device independent quantum money and providing impossibility proof for fully device independent money schemes.
In Section \ref{sec:gresham} we discuss a possible quantum analogue of the Oresme-Copernicus-Gresham Law. We conclude in Section \ref{sec:discusion} by comparing our scheme to the existing ones, discussing the technological difficulties in possible implementation, and summarizing the paper with some interesting open problems. Additionally, in Appendices \ref{app:def}, \ref{app:lem}, and \ref{app:main} we present a rigorous security proof of the scheme, briefly describe honest implementation in Appendix \ref{aap:honest}, and discuss the amount of required memory in Appendix \ref{app:mem}.

\section{Previous works}\label{sec:previous}

An idea of quantum money proposed by Stephen Wiesner was to our knowledge the first application of the quantum effects to the information theoretic, in fact cryptographic task. In this section we will discuss the previous research in this topic using division into private and public key quantum money suggested by Aaronson \cite{DBLP:journals/eccc/Aaronson16, Aaronson2009}. In the private key quantum money schemes only the mint itself can verify the banknote. On the other hand, in the public key quantum money schemes anyone can verify the banknote using publicly available verification procedure, but still no one, except the mint, cannot copy or create new banknote. We will conclude by giving (in Figure \ref{fig:monies}) a comprehensive comparison of different classical and quantum moneys scheme together with their security assumptions.

\begin{figure*}
	\begin{tikzpicture}[
	grow                    = right,
	edge from parent/.style = {draw, -latex},
	every node/.style       = {font=\footnotesize},
	sloped,
	level 1/.style={sibling distance=14em, level distance=4em},
	level 2/.style={sibling distance=8em, level distance=6em},
	level 3/.style={sibling distance=3em, level distance=16em},
	level 4/.style={sibling distance=6em, level distance=16em}
	]
	\node [troot] {Money}
	child {
		node [tenv, bottom color=yellow!20] {quantum} 
		child { 
			node [tenv, bottom color=yellow!20] {private-key} 
			child { 
				node [tenv, bottom color=yellow!20] {\st{device independent}} 
				child {node [tleaf, bottom color=yellow!20] {$\mathbf{This \text{ } work}$ (See \ref{sec:protocol})} edge from parent}
				edge from parent
			}
			child { 
				node [tenv, bottom color=yellow!20] {semi-device independent} 
				child {node [tleaf, bottom color=yellow!20] {$\mathbf{This \text{ } work} $ (See \ref{sec:sdi})} edge from parent}
				edge from parent
			}
			child { 
				node [tenv] {measurement-device independent} 
				child {node [tleaf] {See section \ref{sec:privatekey}} edge from parent}
				edge from parent
			}
			child { 
				node [tenv] {device dependent}
				child {node [tleaf] {See section \ref{sec:privatekey}} edge from parent}
				edge from parent
			}
			edge from parent
		}
		child { 
			node [tenv] {public-key} 
			child { 
				node [tenv] {delocalized} 
				child {node [tleaf] {Quantum Bitcoin (See \ref{sec:publickey})} edge from parent}
				edge from parent
			}
			child { 
				node [tenv] {localized}
				child {node [tleaf] {See section \ref{sec:publickey}} edge from parent}
				edge from parent 	
			}
			edge from parent
		}
		edge from parent
	}
	child { 
		node [tenv] {classical} 
		child {
			node [tenv] {digital}  
			child { 
				node [tenv] {cryptocurrencies} 
				child {node [tleaf] {Bitcoin, Ethereum, \ldots} edge from parent}
				edge from parent
			}
			child { 
				node [tenv] {centralized} 
				child {node [tleaf] {CC, online banking, transfers} edge from parent}
				edge from parent
			}
			edge from parent
		}
		child { 
			node [tenv] {physical} 
			child { 
				node [tenv] {fiat money} 
				child {node [tleaf] {USD, EUR, CNY, PLN, \ldots} edge from parent}
				edge from parent 	
			}
			child { 
				node [tenv] {commodity money} 
				child {node [tleaf] {Gold, sliver, salt, \ldots} edge from parent}
				edge from parent	
			}
			edge from parent
		}	
		edge from parent
	};
	\end{tikzpicture}
	\caption{%
		Types of moneys and their security against forgery: 
		commodity money security is based only on its high intrinsic value.
		Fiat money's security depends mostly on secret products and procedures used in the money making process. For example the paper recipe or the paint chemical composition is kept secret for banknotes. It is worth mentioning that this is against Kerckhoffs's principle.
		Digital money security follows from hardness assumptions for some computational, probably NP-intermediate, problems. In practice RSA algorithm and blockchain techniques are used for which effective attacks, using quantum computer, were proposed.
		Also in case of the public-key quantum money computational assumptions are necessary, but it is unclear if any concrete scheme remains secure.
		Finally the private-key quantum money are information-theoretic secure without any hardness assumptions. Nevertheless we have to consider real life implementation using specific, possibly untrusted, hardware and software. 
	}
	\label{fig:monies}
\end{figure*}
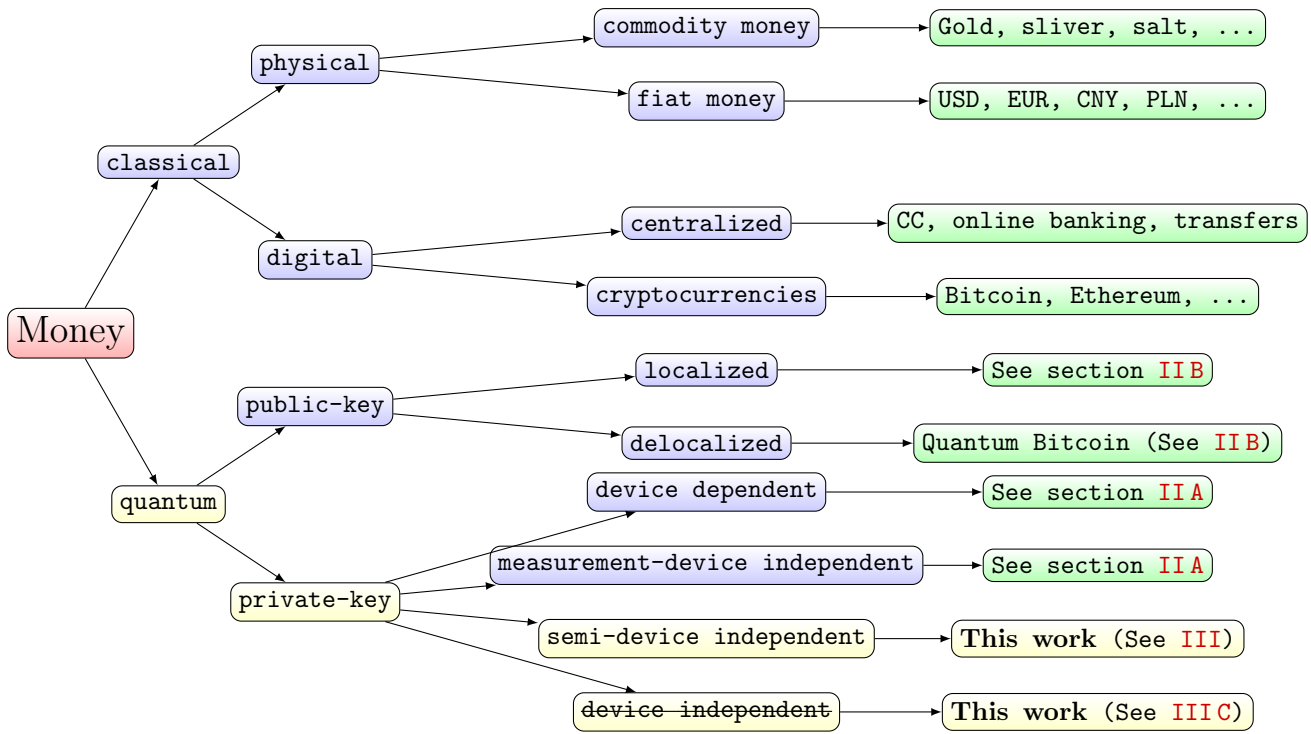

\subsection{Private key quantum money}\label{sec:privatekey}

Around 1970 Stephen Wiesner suggested the first scheme of unforgeable quantum money. Unfortunately his paper was rejected few times and finally was published in 1983 \cite{Wiesner1983}. Even though Wiesner claimed that the protocol is unconditionally secure, a full proof for the most generalized attacks was presented by Molina et al.\ in 2013 \cite{Molina2013}. 

Because of the fact that the scheme requires the mint to maintain a huge database for all produced bills, Bennet et al.\ \cite{Bennett1983} proposed a modification of the protocol, using a cryptographic pseudorandom function, to decrease needed amount of the memory. The question if it is possible to reduce the database size without imposing any computational assumptions was analyzed by Aaronson \cite{DBLP:journals/eccc/Aaronson16}. Later he formally proved that the answer is negative and stated so called Tradeoff Theorem for Quantum Money \cite{1711.01053}. 

Although the above schemes are secure in a regime in which the mint destroy the banknote after verification, allowing to retrieve verified bill is dangerous. So called interactive attacks were independently proposed by Aaronson \cite{Aaronson2009} and Lutomirski \cite{1010.0256}. Even more sophisticated version of interactive attack, based on idea of Elitzur-Vaidman bomb tester, was later suggested by Nagaj et al.\ \cite{Nagaj:2016:AAW:3179330.3179337}. 

The above mentioned scenarios requires visiting the mint, or at least having secure quantum channel so Gavinsky suggested version of quantum money with classical verification \cite{Gavinsky2012}. It is important to notice that in Gavinsky's scheme the Bank does not need to trust the measurement device. Such type of schemes are commonly called measurement-device independent. Another similar scheme was also presented by Georgiou and Kerenidis \cite{georgiou}.

Additionally Pastawski et al.\ \cite{Pastawski2012} and more recently Amiri and Arrazola \cite{Amiri2017} analyze more realistic scenario in the presence of noise and errors. 

It is also worth to mention about fundamentally different approaches aimed at anonymity. Mosca and Stebila \cite{AMS-MosSte10}, (see also Tokunaga et al.\ \cite{AQCash}) proposed quantum coins in such a way that all coins are identical. Their scheme uses black box model that makes thorough security analysis difficult.

Furthermore Selby and Sikora \cite{qptmoney} analyzed unforgeable money in the Generalized Probabilistic Theories.

It is important to point out that also experimental results in quantum money field were presented by three groups of Bartkiewicz et al.\ \cite{Bartkiewicz2017}, Bozzio et al.\ \cite{Bozzio2018} and Guan et al.\ \cite{Guan2018} respectively. Although the theoretical schemes are secure, in the case of real life implementation new vectors of attack could appear. For example Jiráková et al.\ \cite{1811.10718} show how implementation by Bozzio et al.\ \cite{Bozzio2018} could be attacked.

Finally, soon after the first version of this paper was published on ar$\chi$iv preprints repository, Bozzio et al.\ \cite{1812.09256} presented the result with similar title ``Semi-device-independent quantum money with coherent states'' \footnote{which we would prefer to call ``Partially-device independent'' since the name of ``Semi-device independent'' already have well-established position in quantum information field when referring to schemes similar to the one by Paw\l{}owski and Brunner \cite{sdi}.}. They result requires stronger security assumptions but is more focused toward realistic implementations.

\subsection{Public key quantum money}\label{sec:publickey}

The biggest drawback of all private key quantum money schemes is that only the mint can verify the bill. To get rid of this problem, an idea of a public key quantum money, was invented. In that approach not only the mint, but anyone, even untrusted party, could verify the quantum banknote without communication with the mint. General formulation of the public key quantum money was presented by Aaronson \cite{Aaronson2009} and later it security was analyzed by Aaronson and Christiano \cite{Aaronson2012}. 

Following these seminal results many candidates for the private key quantum money scheme was presented. The first such scheme, based on stabilizer states, was proposed by Aaronson \cite{Aaronson2009} but it was later broken by Lutomirski et al.\ \cite{LutomirskiAFGKHS10}. There were also some attempts exploring an idea of local hamiltonians problem that can also be broken using a single-copy tomography presented by Farhi et al.\ \cite{Farhi2010}. Another idea, based on knot theory, was proposed by Farhi et al.\ \cite{Farhi2012}. It remains unbroken but there is no full security proof. 

Until now more papers concerning the public key quantum money or an analysis of its security was published that we should point out here \cite{2011arXiv1107.0321L, Aaronson20122, Pena2015}.

In the most recent work, Mark Zhandry \cite{2017arXiv171102276Z} proved that if the injective one-way functions and a indistinguishability obfuscator exist, then the scheme of the public key quantum money exists. Furthermore he shows how to adapt the Aaronson and Christiano's scheme \cite{Aaronson2012} using these assumptions to get the secure public key quantum money.

We should also mention an ongoing research on decentralized quantum currencies. First Jogenfors \cite{2016arXiv160401383J} proposed Quantum Bitcoin that connects ideas of quantum money and classical blockchain system like the one used in Bitcoin. Later Ikeda \cite{2017arXiv170804955I} presented another approach called qBitcoin based on quantum teleportation and a quantum chain, instead of the classical blocks. Also a cryptocurrency called qulogicoin, based on another version of the quantum blockchain, was also proposed by Sun et al.\ \cite{qulogicoin}.
Recently Adrian Kent proposed a concept of ``S-money'' \cite{1806.05884} and Daniel Kane created a new money scheme based on modular forms \cite{1809.05925}.
Finally Andrea Coladangelo proposed decentralized, blockchain based hybrid classical-quantum payment system \cite{1902.05214} strengthening quantum lightning scheme.

\section{Semi-device independent quantum money}\label{sec:sdi}

In this section we first demonstrate simple attacks on some of the private key
quantum money schemes, including Wiesner's and Gavinsky's ones, which are based on the \textit{cooperation of the mint and the counterfeiter}. Next we discuss why it is impossible to make fully device independent money scheme in Sec. \ref{sec:DI}. We then describe the scheme for semi-device independent quantum money in Sec \ref{sec:protocol}. Next we compare our money scheme with the corresponding SDI QKD protocol of Paw\l{}owski and Brunner, that we use as a base (see Sec. \ref{sec:compare}). Finally in Sec. \ref{sec:proof} we show the idea of the proof, details of which are presented in the appendices.

\subsection{Simple joined attacks: when the mint and counterfeiter collaborate}

We aim to demonstrate that both the original Wiesner scheme and that of Gavinsky are vulnerable to the joined attack. Moreover the attack is general enough to apply to other private quantum money schemes, as it bases on dropping important security assumption: the privacy of the key. Before presenting the attacks we recall how a honest source prepare the state:
\begin{equation}
\begin{split}
&\sum_{(b,v)\in \{0,1\}^{2k}} \frac{1}{4^k} \ket{b,v}\bra{b,v} \xrightarrow{\text{honest}} \\
&\sum_{(b,v)\in\{0,1\}^{2k}} \frac{1}{4^k} |b,v\>\<b,v|\otimes \ket{\rho^W_{(b,v)}}\bra{\rho^W_{(b,v)}}
\end{split}
\end{equation}
 where $\rho^W_k \in\{ \ket{0}, \ket{1}, \ket{+}, \ket{-} \}$, the bit-string $b$ tells the (random) choice of basis, while $v$ corresponds to outcomes. In the original money only the system $W$ contains banknote's state.
 
Now we are ready to show here three attacks of different types. 
The first one enlarges the memory of the banknote, the second one uses additional entanglement and the third makes it a classical state. The first reduces to simple imprinting of the secret key of the Bank directly in banknote's state. This is at the expense of enlarging dimension
of its quantum memory:
\begin{equation}
	\begin{split}
	&\sum_{(b,v)\in \{0,1\}^{2k}} \frac{1}{4^k} \ket{b,v}\bra{b,v} \xrightarrow{\text{attack}_1} \\
	&\sum_{(b,v)\in\{0,1\}^{2k}} \frac{1}{4^k} |b,v\>\<b,v|\otimes \ket{\rho^W_{(b,v)}}\bra{\rho^W_{(b,v)}}\otimes \ket{b}\bra{b}_{H}
	\end{split}
\end{equation}

The mistrustfully prepared banknote has an additional ``hidden'' register $H$ enabling the attack. This register can be used to generate unlimited number of identical banknotes via repetitive von Neumann measurement of system $W$ in the basis indicated by vector $|b\>\<b|_H$. Allowing for such a strong attack, one can imagine that in principle the whole string $|b,v\>$ could be also imprinted in money's memory at a price of doubling it, however imprinting $|b\>$ is enough.
Operations of copying such a ``banknote'' can pass unnoticed from the point of view of the honest Client. From this trivial example we have then learned that in the case of an unbounded dimension of the banknote, its security against forgery is compromised.

The second attack does not require extra memory in the banknote's state but makes use of an additional entanglement between adversary and the untrusted source device. 
Instead of preparing one from the four honest states \(\{ \ket{0}, \ket{1}, \ket{+}, \ket{-} \}\) the source performs locally one from the four Pauli unitaries on a half of a singlet and sends it as a money state. Later the adversary performs global Bell measurement on whole two qubit system. That procedure based on superdense coding \cite{Bennett1992SDC} allows him to obtain both parts of the secret key and prepare arbitrary number of valid banknotes.

In the third one, the mint and the counterfeiter can attack jointly without increasing the memory of the banknote and furthermore, without any additional entanglement, by using only classical states (diagonal in a single basis):
\begin{equation}
	\begin{split}
	&\sum_{(b,v)\in\{0,1\}^{2k}} \frac{1}{4^k} |b,v\>\<b,v| \xrightarrow{\text{attack}_3}\\ &\sum_{(b,v)\in\{0,1\}^{2k}} \frac{1}{4^k} |b,v\>\<b,v|\otimes |v\>\<v|_H.
	\end{split}
\end{equation}
In each run, right before the measurement is physically done, the measurement device  
is given the type of basis $b$ taking value $0$ in case of $\{|0\>,|1\>\}$ and $1$ for  $\{|+\>,|-\>\}$. It can then safely output the value $|v\>\<v|_H$ as a good answer. The two bits that cannot be encoded in $1$ qubit are split into measurement type (revealed later), and its outcome. It is important to notice that this attack will not work in Wiesner's scheme since there the Bank makes measurement itself, but it will work in some schemes money schemes with classical verification.

The scheme of money that we propose (see Section \ref{sec:protocol})  bases on the semi-device independent quantum key distribution protocol that matches as a partial countermeasure to  these three attacks. In the latter protocol one assumes that there are only qubits sent, so the first attack (by enlarging memory) is not applicable. On the other hand, SDI QKD protocol gets accepted only if the data coming from quantum states is observed, i.e., that the
systems communicated were not classical bits, disabling thereby the second attack.
This, and the fact that the honest implementation of the quantum states 
processed by the parties in the SDI QKD are Wiesner's money, motivates
us to study security of Wiesner's scheme under the verification of the SDI QKD protocol (as we describe in detail in Section \ref{sec:protocol}). Before that, in the next section, we will present additional result, no-go for Device-Independent quantum money which highlights the importance of our money scheme.

\subsection{Impossibility of fully device independent quantum money} \label{sec:DI}

	In this section we will show that it is impossible to create a fully device independent money scheme. We will prove this in a scenario where the Bank has at least two branches that do not communicate during the verification phase. In device independent approach both source and measurement devices are untrusted and there are no restrictions on state dimension or additional entanglement, opposed to the semi-device independent approach that we will present in the next sections. On the other hand, we allow all Bank's branches to have shared randomness that can be used both in the state preparation and verification phases. We also assume that the no-signaling condition is fulfilled and post processing is honest, which is the standard approach in most device independent protocols.

	\begin{observation}[No-go for device independent quantum money]
	\label{obs:nogo}
	It is impossible to create fully device independent money scheme with untrusted source and measurement devices that could be produced by adversary.
\end{observation}

Intuitively it is easy to see that using an appropriate modification of the first attack from the previous section any money scheme can be broken. Indeed, without communication between Bank's branches, a malicious mint can always prepare two copies of the banknote in such a way that both will pass verification in different branches. The only way a branch can verify the banknote is to check correlations with client's banknote. It is impossible to ensure that the verification will influence, or give knowledge about, correlations of another branch with different malicious copy of the banknote. 
To justify these intuitions we will provide more formal proof of the above theorem in Appendix \ref{app:nogo}.

Because of the impossibility result we can ask how close one can go toward the device independent approach. We partially answer this question by providing in the next section our main contribution, the semi-device independent quantum money scheme. It requires weaker assumptions than any previous one. We prove its security against a wide class of important attacks and conjecture that it is also secure in the general case.

\subsection{Semi-device independent quantum money protocol}

\label{sec:protocol}
Motivated by the fact that \textit{joined attacks} can compromise the security
of some private quantum money schemes we will show a partial solution to this problem. In this Section we present a scheme for semi-device independent private key quantum money. The concept of a semi-device independent quantum key distribution was discovered by Paw{\l}owski and Brunner \cite{sdi}. In that scheme the sender does not have to trust neither the source nor the measuring device. Instead, the nontrivial assumptions are that the states sent to the receiver have limited dimension and are disentangled from adversary. See Figure \ref{fig:sdiOrigin}.

\begin{figure}
	\includegraphics[width=1\linewidth, angle=0, trim={1cm 9.8cm 8cm 2cm}, clip]{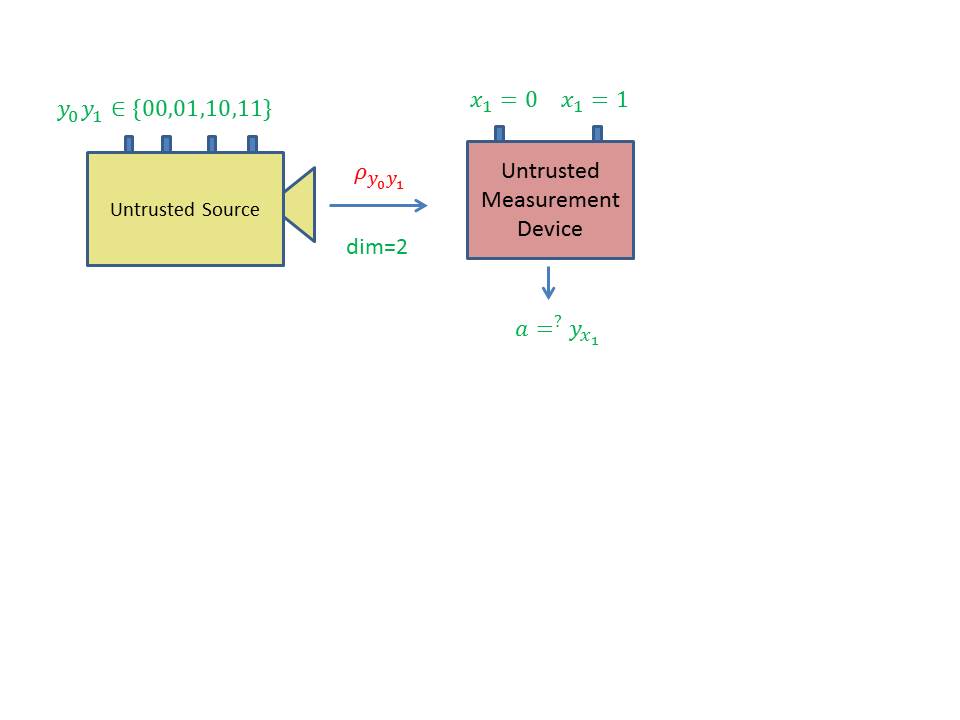}
	\caption{Scheme of the semi-device independent quantum key distribution \cite{sdi} (SDI QKD).}	
	\label{fig:sdiOrigin}
\end{figure}

Our scheme of money will be based on the SDI QKD scheme with the assumption
that the dimension of each state send from sender to the receiver is 
a qubit  ($d=2$). In order to introduce both the concept and a notation it is instructive to recap briefly the semi-device independent quantum key generation protocol \cite{sdi}. The key is produced as follows.
The sender sets up $n$ pairs of random bits $(y_0^i,y_1^i)_{i=1}^n$. In each run of the experiment $i \in [n] \coloneqq \{ 1, \ldots, n \} $, upon pressing the correct
button sender's device produces an untrusted state $\rho_{y^i_0,y^i_1}$,
which is \textit{assumed to be a qubit}, and sends it to the receiver.
Receiver's device is fully untrusted. It measures the state in an arbitrary manner (perhaps knowing state's preparation), yet upon a (random) input $x_i$ it has to output a bit $a_i$ which equals $y^i_{x^i}$. In the classical case, the success probability of guessing the bits of sender is only $3/4$, while in quantum case it is $P_Q \coloneqq \cos^2(\pi/4) \approx 0.8536$. If the guessing probability is larger than a certain value,
the secure key can be established. 

In the SDI quantum money scheme the branches of the Bank play the roles of senders, while the client Alice is the receiver. 
	
	\textbullet \ \textbf{Creation of a single banknote.}
	To create the money all \(k\) branches of the Bank have to posses a common secret randomness that is later stored in classical memories of the branches. Each portion of the bits $(y_0^i,y_1^i)_{i=1}^n$ of this key is attached to some serial number of a separate banknote $SN$ in advance. (Note that the secret key can be obtained for example by measurement on the shared \(2n\) GHZ states \cite{GHZ} or by encrypted classical communication).
	To generate a quantum state of the banknote  associated to the number $SN$  one branch \(B_S\) (in practice the closest to Alice) uses $(y_0^i,y_1^i)_{i=1}^n$ associated with this $SN$ as a sequence of inputs to its untrusted device $S$ (source). The latter device in turn generates $n$ qubits ($\rho_{y_0^i,y_1^i}$) that together form the quantum state of the banknote:
	\begin{equation}
		\bigotimes_{i=1}^n \rho_{y_0^i,y_1^i}.
	\end{equation} 
	The above state is sent to Alice's wallet (dedicated quantum memory device). In the end the joined state of $k$ branches of the Bank and Alice's wallet 
	takes the form:
	\begin{equation}
		\sum_{y=(y_0^i,y_1^i)\in\{0,1\}^{2n}} |y\>\<y|_{B_1}\otimes...\otimes|y\>\<y|_{B_k}\otimes \left(\bigotimes_{i=1}^n \rho_{y_0^i,y_1^i}\right).
	\end{equation}
	
	\textbullet \ \textbf{Verification at the Bank.}  
	 	 Alice comes to any branch $B_l$. The \(B_l\) generates
	 	 a bit-string $(x^i)^n_{i=1}$, inputs the bits to the
	 	 untrusted terminal $T$, and collects the output bit-string $(a^i)_{i=1}^n$.
	 	 For a total data represented by a string of tuples:
	 	 $S_A = (y_0^i,y_1^i,x^i,a^i)_{i=1}^n$
	 	 the Bank accepts it if the following condition is satisfied:
	 	 \begin{equation}
		 	 X_A (S_A) \coloneqq \left|\left\{ i  \in [n] : a^i = y^i_{x^i}\right\}\right| \geq \theta n
		 	 \label{eq:acceptance}
	 	 \end{equation}
	 	 i.e., the number of correct guesses is above the \textit{threshold value} 
	 	 $\theta n$, and rejects otherwise.
	
	\textbullet \  \textbf{Verification at a distance.}  
	 Alice establishes an \textit{authenticated connection for classical communication} with some (arbitrary) branch \(B_l\) of the Bank. The \(B_l\) gives her random inputs $(x^i)_{i=1}^n$, that she should use together with her quantum state from the memory of her wallet as inputs to the untrusted terminal (her own, or, e.g., the one operated by a seller in a shop). The classical output $(a^i)_{i=1}^n$ from the device (possibly modified by Alice to $(a'^i)_{i=1}^n$) is then sent to \(B_l\) that checks if the data $(y_0^i,y_1^i,x^i,a'^i)_{i=1}^n$ are acceptable if inequality (\ref{eq:acceptance}) holds, and rejects it otherwise.
  
For the sake of clarity the whole process of the creation and the verification of the semi-device independent quantum money is illustrated in Figure \ref{fig:sdi} (for the general scheme with many branches) and Figure \ref{fig:sdi2} (for the creation and verification at the same branch). We state below certain remarks on the variants of the above approach.

\begin{figure}
%
%
%
%
%
%
%
%
%
%
%
		\includegraphics[width=1\linewidth, angle=0, trim={2cm 5cm 3cm 5cm}, clip]{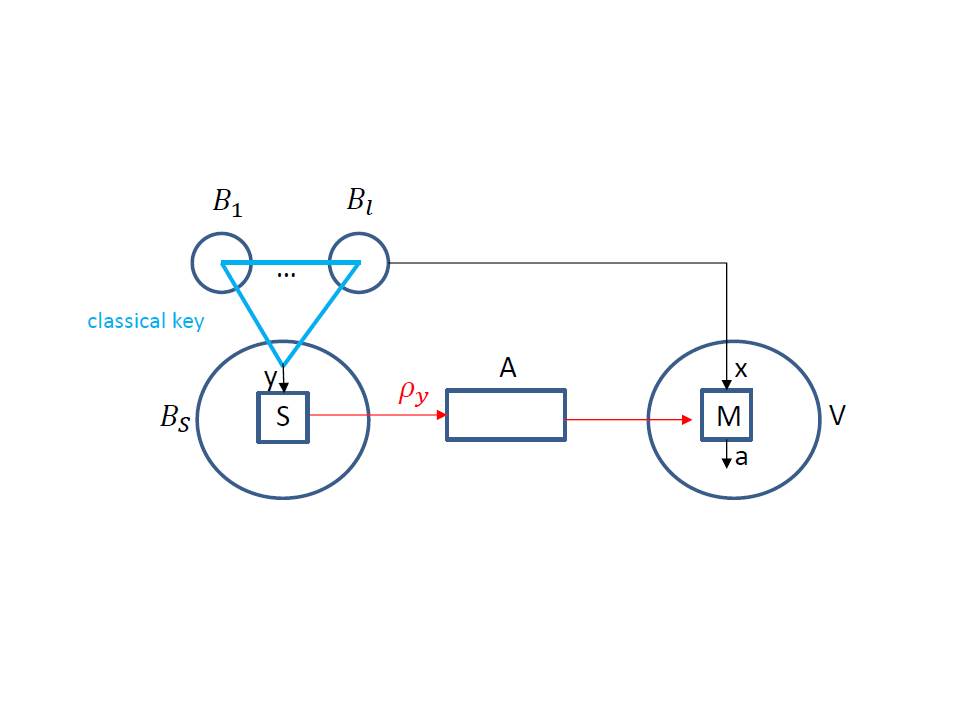}
	\caption{%
		\( B_1 \ldots B_k \) represents arbitrary number of Bank's branches that share common randomness \( y = (y_0^i, y_1^i)_{i=1}^n \) stored in the classical memories. Branch \( B_S \) inputs the string \( y \) into untrusted source device \( S \) and sends the generated \(n\) quantum systems \( \rho_{y} = \bigotimes_{i=1}^n \rho_{y_0^iy_1^i} \) to Alice's memory. When Alice wants to verify the money she visits some branch \( B_l \). This branch generates random binary string \( x = ( x^i )_{i=1}^n \) of length \( n \) and feeds as an input to untrusted measurement devices of the terminal \( M = (M^i)_{i=1}^n \), which generates a string \( a = (a^i)_{i=1}^n \) . The branch then estimates the probability distribution \( P(a|xy) \) and accepts the banknote as valid or rejects it dependently on whether the condition Eq (\ref{eq:acceptance}) is met.
	}
	\label{fig:sdi}
\end{figure}

\begin{figure}
	\includegraphics[width=1\linewidth, angle=0, trim={1cm 0cm 0cm 0cm}, clip]{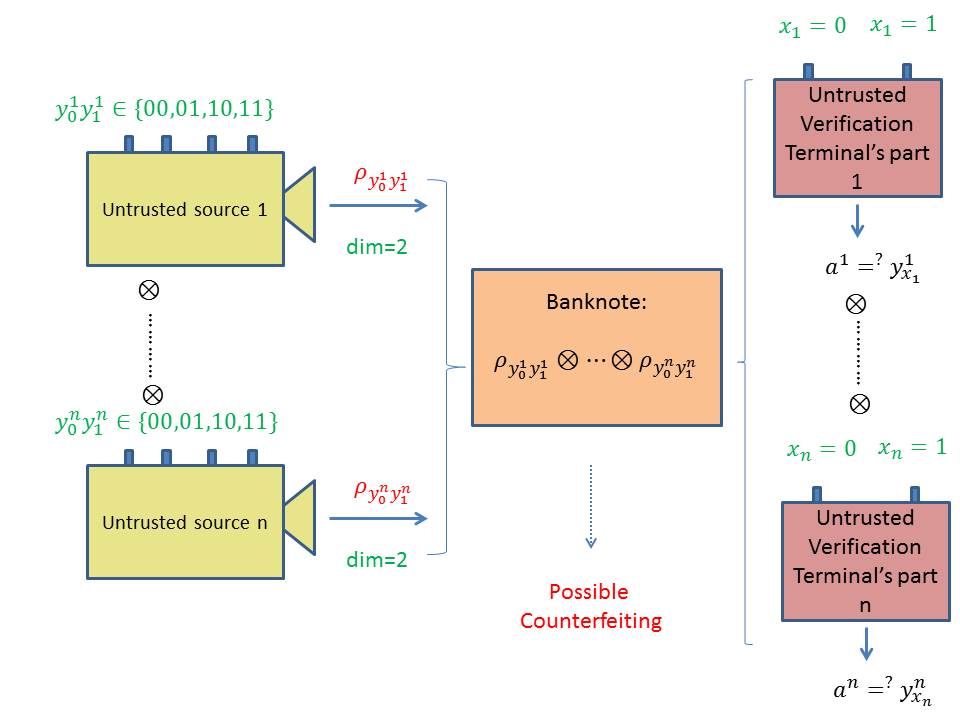}
	\caption{%
		The procedure of generation and verification of a single banknote in the SDI money scheme. The $n$ untrusted source devices independently produce qubit states $\rho_{y_0^iy_1^i}$ that form in total the banknote that is kept in Alice's wallet and is exposed to the counterfeiting by her or Frederick (or even both). The verification of the banknote is done by $n$ independent (not necessarily identical) untrusted parts of a verification terminal, each checking if $a^i = y_{x_A^i}^i$ (i.e., if the its output is equal to one of the two bits of Bank's note at run $i$ chosen randomly as $x_A^i$). The banknote gets accepted if the number of correctly guessed bits exceeds 
		$\beta n$ with $\beta = 2 P_Q\left(1/2 +\eta\right)+M\left(1/2 +\eta\right) +2\eta$, where $\eta$ depends on $n$ (taking care of possible fluctuations of the number of guesses). 
	}	
	\label{fig:sdi2}
\end{figure}

\begin{remark}[The creation of the banknote without communication]
	The branches can create the money without the communication. Using synchronized clocks, they can continuously generate new random inputs. When a client arrive, the serial number of her banknote would contain time that uniquely indicates what remembered randomness verifying branch should use. Additionally branches should agree among themselves on the allowed generation time to make sure not to generate two bills from the same randomness. 
	Similar idea can be also implemented in some previous money schemes, for example Wiesner's one.
\end{remark}

\begin{remark}[Reusing randomness during a reverification]
	Frederick can adopt a strategy to try verify the same banknote many times, waiting for such a random string from Bank's branch. For example he would like to have the size of set \( D(\bm{x}_\oplus) \) from the proof (see Appendix \ref{app:lem}) as big as possible the size of which would be providing verification of the banknote. In order to prevent this, the verifying branch, after the first attempt of the verification, could store the verification randomness in the memory and use it in all next tries of the verification of that banknote.
\end{remark}

\begin{remark}[Predefined agreement on the queries of the branches]
	Instead of using an independent randomness by all branches during the verification, the branches could earlier agree on some verification string for all banknotes. It is possible to do it in similar manner like in the banknote generation procedure. Although it is not necessary in order to maintain the security, this approach could reduce complexity of the proof and decrease a number of the qubits required in the protocol.	
\end{remark}

\subsection{Comparison of the SDI QKD and SDI money scheme}

\label{sec:compare}
We make now an explicit comparison of our protocol with
that of the semi-device independent of Paw\l{}owski and Brunner \cite{sdi}. There are three main differences.

	\textbullet \ \textbf{Memory requirements:}	
	A conceptual difference is that  we defer the process of measurement and call the states prepared by the source in SDI protocol collectively the banknote. The process of the measurement is identified by us with the verification done by the terminal at some later time. It is of particular convenience that the SDI protocol
	does not rely on the no-signaling principle, so the measurement of the
	banknotes can be done any time after they were prepared.
	In other words, our protocol needs the quantum memory while SDI does not.
		
	\textbullet \ \textbf{Limited number of runs:}	
	A significant practical difference is that SDI quantum money scheme corresponds to a limited SDI QKD protocol to the creation and the verification procedures without the privacy amplification and information reconciliation part. In particular, in our protocol the number of runs of the corresponding SDI QKD experiment (i.e., the length of the banknote) is only long enough to enable estimation of the guessing probability which depends only on the possible systematic error in the experiment and the concentration property due to law of large numbers. This is in contrast with the SDI QKD protocol which involves as many runs (at least) as the number of key bits are needed to be generated. Indeed, we do not aim at creating the secret key, because - there is no phase of the public reconciliation and the privacy amplification. Preparing and verifying
	long key is equivalent to creating and verifying a huge number of banknotes.
	
	\textbullet \ \textbf{An intermediate acceptance threshold:}	
	The third difference concerns the acceptance threshold. Acceptable range of the value of the probability of guessing $P_{\text{guess}}$ of the string $(y_0^iy_1^i)_{i=1}^n$ in the SDI QKD protocol varies from the maximal $P_Q\approx 0.8536$, which implies the highest possible key rate in this scenario, to the minimal $P_{\text{crit}}^{\text{key}}\approx 0.8415$, which implies zero key rate. Let us stress here, that \textit{any} value between $P_Q$ and $P_{\text{crit}}^{\text{key}}$ is acceptable, as leading to a non-zero key rate (yet, one aims at the highest). Instead, in the corresponding SDI money scheme one needs the value of this parameter to be larger than $P^{\text{money}}_{\text{crit}} \coloneqq (P_Q + P_{\text{crit}}^{\text{key}})/2\approx 0.84755$. On the other hand, all money schemes with the acceptance threshold $\theta$ in the range $(P^{\text{money}}_{\text{crit}},P_Q]$  are protected against forgery given large enough number of qubits of the banknote $n$.

\subsection{Security proof of the SDI money scheme}

\label{sec:proof}
In this section we provide the proof of the main result: the SDI quantum money scheme is protected against the \textit{qubit-by-qubit forgery}. That is, against the case when the mint and the counterfeiter (as well as the verification terminal possibly created by the counterfeiter) cooperate in a manner
that each qubit is attacked (prepared, copied and tested) independently.
Under some additional necessary assumptions, which we list below, we show that two, cooperating clients, Alice and Frederick, cannot get the banknote accepted as valid in \textit{two} Bank's branches. As we show, the case of \textit{many} Bank's branches follows from the security in the latter case. The case of a \textit{birthday attack} of choosing best pair of branches is then taken care of by the union bound. Indeed let us assume that number of branches equals $k= \mbox{poly}(n)$, where $n$ is the length of the banknote, which is a reasonable constraint possible to be satisfied. If for any pair the probability of successful counterfeiting is exponentially small $\epsilon_{2}(n) \sim O(e^{-n})$, the highest probability of this event for $k$ branches is not higher than $\epsilon_k \coloneqq \binom{k}{2} \epsilon_2(n)$, which is still small (of order $O(e^{-n})$). 

Before we will present the proof, we will first state explicitly all more or less implicit assumptions of our scheme. Note, that first two are necessary in general (not only in the SDI money scheme).
\begin{enumerate}[label=\textbf{ASM\arabic*}] \label{ite:assumtions}
	\item \label{asm:1} Bank's branches have access to a private fully random number generator that they use to generate \(y\)'s and \(x\)'s.
	\item \label{asm:2} Branches of the Bank use an honest classical post-processing units in the verification procedure.
	\item \label{asm:3} The dimension of the state that is produced at the output of the source (i.e., mint) is bounded and there is no other information leaking from the source to Alice or Frederick.
	\item \label{asm:4} The state produced by the source is unentangled from the dishonest parties (Alice and Frederick).
	\item \label{asm:5} The source devices create the states in independent way, what also implies that each of the sources have access only to the its input (not the inputs of the other sources).
	\item \label{asm:6} The measurement devices are independent, each measure only its subsystem, and the outputs of Alice and Frederick in each run are independent from the inputs and the outputs from another runs.
\end{enumerate}

In particular case of the presented SDI money scheme, in ASM3 we specify that
each of the independent parts of the sources (as specified in ASM5) has
output bounded by $d=2$, i.e., the source works by producing independently $n$ qubits
(however not necessarily in the same way).

\begin{remark}[On the possible weakening of  the assumptions]
	It seems plausible that the assumption ASM6 could be omitted but it would complicate  the proof. The question if we can omit the assumption ASM5 is a hard open problem, related to the formulation of the SDI QKD scheme and Random Access Codes \cite{RAC1,RAC2} in general. On the other hand, all other assumptions are necessary to prove the security of our scheme since rejecting any of them leads to a successful attack.
\end{remark}

Let us briefly describe the idea of the proof of the security of the scheme. It is a consequence of two facts: (i) the monogamy inherent to the SDI key generation protocol and (ii) the fact that each Bank's branch queries independently from the other branches during the verification procedure. It will hold for the case when the Bank
verifies the banknote via untrusted terminal, i.e., Alice (and / or Frederick) come
to the Bank to get the banknote accepted. The case with the communication is then reduced to the latter, under assumption that the strategy to lie about the outputs
of the devices (which is then at a choice of the dishonest parties) is individual, independent for each of the runs of the protocol (see Appendix \ref{app:def}). As we discuss in Section \ref{sec:discusion}, this a bit unrealistic assumption that can be in principle dropped given the protocol of SDI QKD is proven to be secure against the general, so called \textit{forward signaling attacks}.

\subsubsection{The case of attack on a single qubit}

By definition, much like in the Wiesner scheme, for a banknote to be accepted, its owner has to guess correctly the bits of the Bank. To see that two dishonest persons, Alice and Frederick, can not both pass the verification of our banknote it is instructive to focus on the attack on a single qubit of the banknote. Suppose Alice and Frederick are trying to ``split'' its use to maximize the probability of guessing $P_{guess}$ in two experiments of some two branches of the Bank. Their joined attack can be described as a conditional probability distribution (a box): $P(a_F,a_A|x_A,x_F,y_0, y_1)$, where $y_0$ and $y_1$ are the secret keys of the Bank which Alice and Frederick are trying to guess, $x_A$ and $x_B$ are the random inputs to the box generated by the Bank. 

For simplicity of description we will assume that
Alice and Frederick comes to the Bank, while the Bank who sets the input to the 
devices (we will argue later how to partially relax this assumption).
The joined attack aims at generating two bits $a_A$ (by Alice)
and $a_F$ (by Frederick), such that the probability of guessing 
the $x_A$th bit of $y_B = (y_0,y_1)$ and $x_F$th bit of $y_B$ 
by Frederick are both maximal. The guessing probability for Alice
and Frederick respectively read:
\begin{equation}
	P_{guess}^A= \frac{1}{8} \sum_{y_0,y_1,x_{A}} P(a_A = y_{x_A}|y_0,y_1,x_A),
\end{equation}
and
\begin{equation}
	P_{guess}^F= \frac{1}{8} \sum_{y_0,y_1,x_{F}} P(a_F = y_{x_F}|y_0,y_1,x_F).
\end{equation}
Let us observe first that in the case $x_A=x_F$, they can both
achieve the maximal possible probability of guessing $P_Q= cos^2(\pi/8) \approx 0.8536$ \cite{sdi}. Indeed, Alice can come first to
one branch, and behave honestly having the guessing probability $P^A_{guess}=P_Q$, while Frederick can \textit{copy her answer},
reaching the same probability of guessing. However, when
$x_A \neq x_F = x_A\oplus 1$, the dishonest parties need to
guess opposite bits: $y_0$ (Alice) and $y_1$ (Frederick) or vice versa
(with half probability). However, it is proven in \cite{sdi} that
\begin{equation}
	P_{guess}^{AF}(y_0) + P_{guess}^{AF}(y_1) \leq \frac{5+\sqrt{3}}{4}\eqqcolon M.
\end{equation}
Hence, even if Alice and Frederick were fully collaborating, the sum of the probabilities of guessing of the two bits is bounded.

Since $x_A = x_F$ with the probability one half, averaging over the value of $x_A\oplus x_F$
we conclude that the average number of correctly guessed bits has an upper bound
\begin{equation}
	\frac{1}{2}(2 P_Q + M) \coloneqq B.
	\label{eq:B}
\end{equation}

In what follows we will prove that due to the independent nature of the attack,
the above bound, multiplied by the number of runs $n$, applies (up to fluctuations $\eta$ around the average). The corresponding bound enlarged by the maximal possible fluctuations reads then $n \beta$ with $\beta = 2 P_Q\left(1/2 +\eta\right)+M\left(1/2 +\eta\right) +2\eta$. We will
then choose the threshold value $\theta$ to be larger than ${\beta}/2$. This
will assure, that the two dishonest parties can not get the same banknote accepted in two Bank's branches, as their total sum of the guesses would be larger than $2\beta/2  = \beta$, reaching the desired contradiction.

\subsubsection{Extending the argument to the general case of the qubit-by-qubit attack}

We would like to extend this reasoning to the case of
the repeated experiment of $n$ runs ($n$ will be relatively small,
as short as the length of an usual preamble of the QKD protocols). We assume
here that the attack is ``id'', i.e., by not necessarily equal
however independently distributed random variables, according to the measure:
\begin{equation}
	\label{eq:main-dist1}
	\mu \sim	\bigotimes_{i=1}^n U \left(y^i_0, y^i_1, x^i_A, x^i_F \right) P\left(a^i_A,a^i_F | y^i_0, y^i_1, x^i_A, x^i_F \right),
\end{equation}
where $U(y^i_0,y^i_1,x^i_A,x^i_F)$ denotes the uniform distribution over its arguments. We then observe that, instead of providing
$x_A$ to Alice and $x_F$ to Frederick, the two branches of the Bank
could give $x_A$ to Alice and $x_\oplus \coloneqq x_A\oplus x_F$ to Frederick. This is because Alice and Frederick are collaborating, so they can 
compute back value of $x_F$ from these data in case it is needed.
We can therefore change the scenario to one in which the
parties are given $(x_A,x_\oplus)$, if the probability measure
is changed accordingly to the following one:
\begin{equation}
	\begin{split}
	\mu' \sim& \bigotimes_{i=1}^n U \left(y^i_0, y^i_1, x^i_A, x^i_\oplus \oplus x^i_A \right)\\ &\times P\left(a^i_A,a^i_F | y^i_0, y^i_1, x^i_A, x^i_\oplus \oplus x^i_A \right).
	\end{split}
\end{equation}
The measure \( \mu' \) acts on \( x_A \) and \( x_\oplus \) in the same way as \( \mu \) would acts on \( x_A \) and \( x_F \), so in some sense it is undoing the XOR operation.
This modification of scenario does not change \textit{the probability of
	successful forgery}, i.e., the probability of an event in which both Alice will get accepted as supposed to have a valid banknote and so will happen to Frederick.
To see this, we first note that a set of events (denoted as $\mathcal{F}$) leading to a successful forgery reads:
\begin{equation}
	\begin{split}
	\mathcal{F} \coloneqq \left\{(y^i_0, y^i_1, x^i_A, a^i_A, x^i_F,  a^i_F)_{i=1}^n \in \{ 0,1 \}^{6n}\right.: \\ \left|\left\{ i  \in [n]  : a^i_A = y^i_{x^i_A} \right\}\right| > \theta n,\\
	\left. \left|\left\{ i  \in [n]  : a^i_F = y^i_{x^i_F} \right\}\right| > \theta n \right\}.
	\end{split}
\end{equation}
We will also define a strategy \(S\) by 
\begin{equation}
	S \coloneqq (y^i_0, y^i_1, x^i_A, a^i_A, x^i_F, a^i_F)_{i=1}^n.
\end{equation}
We then prove (see Corollary \ref{cor:accacc'}) that
\begin{equation}
	\begin{split}
	P(\mathcal{F}) \coloneqq& \sum_{S\in \mathcal{F}} P_{\sim \mu}(\mathcal{F})\\
	=& \sum_{S'\in \mathcal{F}'} P_{\sim \mu'} P(S')\\ 
	\eqqcolon& P(\mathcal{F}'),
	\label{eq:equivalence}
	\end{split}
\end{equation}
where $S' = (y^i_0, y^i_1, x^i_A, a^i_A, x^i_\oplus, a^i_F)_{i=1}^n$
and 
\begin{equation}
	\begin{split}
	\mathcal{F}' \coloneqq \left\{(y^i_0, y^i_1, x^i_A, a^i_A, x^i_\oplus, a^i_F)_{i=1}^n \in \{ 0,1 \}^{6n} : \right.\\ \left|\left\{ i  \in [n]  : a^i_A = y^i_{x^i_A} \right\}\right| > \theta n,\\
	\left. \left|\left\{ i  \in [n]  : a^i_F = y^i_{x^i_\oplus \oplus x_i^A} \right\}\right| > \theta n \right\}.
	\end{split}
\end{equation}

Due to the fact that $\textbf{x}_\oplus$ is created from fully random bits that are unknown for adversary during creation of the money, we have
\begin{equation}
	P(\mathcal{F}') = \sum_{\bm{x}_\oplus \in \{0,1\}^n} p(\bm{x}_\oplus) P(\mathcal{F}'|\bm{x}_\oplus).
\end{equation}
We can narrow considerations to the \textit{typical} $\bm{x}_\oplus$, i.e., those having number of symbol $0$ and $1$ approximately $n/2$ times. More formally the set of typical sequences is defined as
\begin{equation}
	\mathrm{\mathcal{T}}_{\eta} \coloneqq \left\{ \bm{x} : \left| \frac{|\bm{x}|_0}{n}  - \frac{1}{2} \right| \leq \eta \right\}, 
\end{equation}
where $|\bm{x}|_0$ is the number of positions with symbol $0$ in a bitstring
$\bm{x}$. All sequences of the length $n$ (given $n$ is large enough) are with high probability typical (i.e., with a probability $1-\epsilon(\eta)$ for $\epsilon(\eta)=2 \exp(-2\eta^2 n)$). We have therefore

\begin{equation}
	P(\mathcal{F}') \leq \sum_{\bm{x}_\oplus \in \mathcal{T}(\eta)} p(\bm{x}_\oplus) P(\mathcal{F}'|\bm{x}_\oplus) + \epsilon(\eta).
	\label{eq:split-to-typ}
\end{equation}

We then see that one can fix a typical $\bm{x}_\oplus$, and 
prove that for any such \( \bm{x}_\oplus \) the probability of acceptance is low.
We will assure it by setting an appropriate $\theta$, so that
with a high probability over the conditional measure $\mu''\coloneqq \mu'(\bm{x}_\oplus)/ p(\bm{x}_\oplus)$ the
strings $S'$ emerging from the attack will be rejected as
having too low number of guessed bits of $(y_0^i,y_1^i)_{i=1}^n$.

In more detail, we first note that for a fixed $\bm{x}_\oplus$, on average over $n$ runs with respect to the measure $\mu''$, there are no more guessed inputs than $n B$ with $B$ given in Eq. (\ref{eq:B}). It remains to take into account the fact, that the observed number of the guessed inputs need not to be  
equal to its average. However $\bm{x}_\oplus$ is typical, hence the number of runs will be at least $n/2 - \eta n$, so we can use that the attack is performed in the independent manner. 
Due to Hoeffding's inequality, we obtain that the observed number of guessed inputs is with a high
probability bounded from above by:
\begin{equation}
	\begin{split}
	n  \beta \equiv& n \left[ (2 P_Q\left(\frac{1}{2} + \eta\right) + M \left(\frac{1}{2}+\eta\right) +2\eta \right] \\  =& n \left[\left( 2 \cos\left( \frac{\pi}{8} \right) + \frac{5 + \sqrt{3}}{4} \right) \left(\frac{1}{2} + \eta\right) +2\eta \right]
	\end{split}
\end{equation} 
where $\eta$ takes care of the maximal possible fluctuations.

Before we explicitly control these fluctuations, we first define four random variables describing the guessing at the $i$th run of the verification procedure by Alice and Frederick as
\begin{equation}
	X_{i,A}^{\bm{x}_\oplus} \coloneqq 
	\begin{cases}
		\delta\left(a^i_A,y^i_{x_A^i}\right) &: x_A^i = x_F^i \\
		0 &: x_A^i \neq x_F^i
	\end{cases},
\end{equation}

\begin{equation}
	X_{i,F}^{\bm{x}_\oplus} \coloneqq 
	\begin{cases}
		\delta\left(a^i_F,y^i_{x_A^i}\right) &:  x_A^i = x_F^i \\
		0 &: x_A^i \neq x_F^i
	\end{cases},
\end{equation}

\begin{equation}
	Y_{i,A}^{\bm{x}_\oplus} \coloneqq 
	\begin{cases}
		\delta\left(a^i_A,y^i_{x_A^i}\right) &: x_A^i \neq x_F^i \\
		0 &: x_A^i = x_F^i
	\end{cases},
\end{equation}

\begin{equation}
	Y_{i,F}^{\bm{x}_\oplus} \coloneqq 
	\begin{cases}
		\delta\left(a^i_F,y^i_{x_A^i \oplus 1}\right) &: x_A^i \neq x_F^i \\
		0 &:  x_A^i = x_F^i
	\end{cases}.
\end{equation}

Now we can get back to describing deviations from the average of the $4$ random variables:
\({\bar X}^A_{\bm{x}_\oplus}, {\bar X}^F_{\bm{x}_\oplus}, {\bar Y}^A_{\bm{x}_\oplus}\) and \( {\bar Y}^F_{\bm{x}_\oplus} \) that are the sums over \( i \in [n] \) for the variables described above. 
The values of \( {\bar X}^A_{\bm{x}_\oplus} ({\bar X}^F_{\bm{x}_\oplus}) \) are the numbers of bits of $(y_0^i,y_1^i)$ correctly guessed by Alice (Frederick) from the positions $i$ satisfying $ x_A^i \neq x_F^i $. Analogously, \( {\bar Y}^A_{\bm{x}_\oplus} ({\bar Y}^F_{\bm{x}_\oplus}) \) describe the number of correct guesses for \( i \) such that $ x_A^i = x_F^i $.
Details are given in Lemma \ref{lem:u-bound} and Corollary \ref{cor:deviation} (see also Appendix \ref{app:lem} for an explicit definition of random variables, their sum and expected values).

The last argument follows from a simple observation. Namely,
if the total fraction of the correctly guessed positions by two persons 
is less than $\beta$, the minimum of the fractions of the correct guesses
by each of them separately is not greater than $\beta/2$. Setting
the acceptance threshold $\theta$ large enough that
the \textit{minimum of the numbers of guesses is below $\theta$}, we assure that for each typical $\textbf{x}_\oplus$ the banknote is rejected with the high probability in at least one branch. In particular, for any $\theta > \beta/2$ this probability is at least $1 - 8\exp(-\eta^2 (n/2 -\eta))$, where for every typical $\textbf{x}_\oplus $ the error $8\exp(-\eta^2 (n/2 -\eta))$ upper bounds the probability of event that at least one of the $4$ random variables \({\bar X}^A_{\bm{x}_\oplus}, {\bar X}^F_{\bm{x}_\oplus}, {\bar Y}^A_{\bm{x}_\oplus}\), \( {\bar Y}^F_{\bm{x}_\oplus} \) is far from its respective average. 

Taking into account Eqs (\ref{eq:split-to-typ}) and (\ref{eq:equivalence}), we obtain finally
\begin{equation}
	P(\mathcal{F}) \leq 5 \epsilon'(\eta),
\end{equation}
with $\epsilon'(\eta) \coloneqq 2\exp(-\eta^2 (n/2 -\eta))$.

\section{Quantum Oresme-Copernicus-Gresham's Law}\label{sec:gresham}

One of the famous laws of economy is:
\begin{itemize}
	\item \textbf{Bad money drives out good.}
\end{itemize}
This law states, colloquially speaking, that if certain money is cheaper to produce,
then it will eventually subside the one that is more expensive to produce, where
expensive is understood not in terms of face value but in terms of intrinsic value.
Although it was named after Sir Gresham, it has been observed by others, even much earlier. The two most cited authors are Nicole Oresme \cite{GreshamOresme} and Nicolaus Copernicus \cite{GreshamCopernicus}, so that the above law is also refereed to as Copernicus', or the Oresme-Copernicus-Gresham's Law. However, the first known appearance of a similar statement is in the comedy ``The Frogs'', written by the Ancient Greek playwright Aristophanes around 405 BC \cite{GreshamArystofanes}. For a overview of the law see e.g.\ \cite{CarolinaSparavigna2014}.

From the perspective of the economy the concept of money is a matter of a social agreement and properties of a given material/procedure used to produce a coin or a banknote. So it might appear that there is no need to consider a quantum variant of Gresham's law per se, because one can apply the OCG Law to the new method of mining -- from quantum states. This is what happens to classical crypto-currencies. Instead, formulating quantum analog of the OCG Law we would like to compare quantum currencies with each other within quantum domain. This is because one could choose a ``cheaper'' way to produce money -- from quantum states that are ``cheaper'' to obtain whatever means ``cheaper'' to quantum technology  at a given moment of development. We therefore would like to introduce and discuss a version of Quantum Gresham's law in the following way:
\begin{itemize}
	\item \textbf{Bad \textit{quantum} money drives out good \textit{quantum} one.}
\end{itemize}

Deciding whether to keep a given quantum currency or not may be a complex process,
depending on various mutually dependent parameters, the importance
of which varies over a change of preferences of particular individuals or societies. It is therefore too early and hence too hard to foresee the behavior of the quantum currency flow between the schemes provided they happen to be realized experimentally. 

\begin{example} 
	For the presented SDI quantum money scheme, to be unforgeable via qubit-by-qubit way, it is enough that the source of the banknote (if it able to produce a large number of banknotes) manages to produce SDI key at a rate $\theta > \beta/2$. It is however not demanded that $\theta \approx P_Q$, i.e., that the source would be able to produce money equivalent to low number of runs of the SDI key generation experiment with the maximal possible rate $P_Q$. 
	
	Suppose that some provider $P_I$ is able to	produce the SDI money which passes the acceptance threshold $\theta =	\beta/2 + 2 \delta$ for some $\delta >0$. Next, suppose that some other provider
	$P_{II}$ is be able to produce a reliable SDI money with 
	the lower acceptance threshold $\beta/2 +\delta$. As we have proved in Theorem \ref{thm:main}, banknotes of both providers are valid and can not be forged under certain assumptions. However the banknote of the provider $P_I$ can be attributed a larger \textit{quantum commodity value} defined as the \textit{SDI key rate of a source which produced the banknote}. From perspective of banknote's holder, this key rate implies nonzero rate of the min-entropy of banknote's quantum state and hence nonzero rate of the \textit{private randomness}. An additional reason to keep the $P_I$ type money and spend more often $P_{II}$ type is that the first one could be more robust to noise. Indeed, even decrease by $\delta$ of the observed fraction of the correct guesses will not invalidate the banknote. 
\end{example}

We make then the following observation:
\begin{observation}
	If the Quantum Oresme-Copernicus-Gresham law applies, the SDI money with lower acceptance threshold $\theta$ would drive out the SDI money with higher $\theta$. (We note here, that we have implicitly assumed that
	the hardware parameters of the realization by $P_I$ and $P_{II}$
	are comparable. Otherwise, realization of the banknote according to type $P_{II}$' receipt can be simply too expensive, e.g. in energy spent on keeping them in a quantum wallet).
\end{observation}

The above example is very limited, as it concerns different
ways of realization of \textit{the same} money scheme (i.e., currency). It is however plausible that if the quantum version of the OCG law turns out to be true, the individuals will tend to keep the most secure, cheapest to produce and to store 
money of the highest commodity value (in the sense of its use for quantum information processing), and will spend the other currencies more often. Going a bit further, one can consider monies in theory $T$ (for such a general approach see \cite{qptmoney}), and have a ``$T$ Oresme-Copernicus-Gresham Law'', a theory-dependent version describing flow of currencies valid in a theory $T$. An interesting
special case would be the ``multi-theory OCG Law'' that could govern the flow of currencies between different sub-theories. A natural example of the latter would be \textit{Classical-Quantum Oresme-Copernicus-Gresham Law}, expressing the behavior of everyday currencies and the quantum ones on the same footing.

\subsection{A comparison of the money schemes}

In the Table \ref{tab:comparison} we present the comparison of different protocols, including original protocol of Stephen Wiesner \cite{Wiesner1983} and that of Dimitry Gavinsky \cite{Gavinsky2012}, and show which (parameters of) quantum devices have to be trusted by the Bank in order to maintain security. 
\begin{table*}
	\begin{tabular}{|l||c|c|c|c|c|}
		\hline 
		\textbf{Protocol} & \textbf{Classical scheme} & \textbf{Wiesner's scheme}\cite{Wiesner1983} & \textbf{Gavinsky's scheme}\cite{Gavinsky2012}  & \textbf{SDI[this work]} & \textbf{DI[this work]} \\ 
		\hline 
		Source & Yes & Yes  & Yes & No & No  \\ 
		\hline 
		Alice's measurement  & N/A & Yes\textsuperscript{1} & No & No & No \\ 
		\hline 
		System dimension & No & No & No  & Yes & No  \\ 
		\hline 
		Number of branches & 1 &  unlimited & unlimited & unlimited & unlimited \\
		\hline
	\end{tabular}
	\caption{%
		Table of the Bank's trust. \textsuperscript{1} The measurement of Alice's subsystem can be performed in Bank's branch or in the partially trusted terminal (for example, in a shop). 
	}
	\label{tab:comparison}
\end{table*}

\subsection{How close we are to practice?}\label{sec:feas}

A fundamental obstacle in realization of the presented one and many other money schemes is the fact that it relies on existence of a reliable and long-time living
quantum memory. 
It is hard to foresee when (if ever) such memories would be available, however 
there are works towards this direction. As an example of a recent huge experimental progress in developing the quantum memory we can invoke paper by Wang et al.\ \cite{Wang2017}, presenting a single-qubit quantum memory that exceeds coherence time of ten minutes. Furthermore Harper and Flammia \cite{IBMErrorCor} demonstrated the first implementation of the error correcting codes on a real quantum computer. This may indicate, that the error correcting codes can become useful in the near future quantum memories. 

We want to emphasize here that the tasks of universal fault-tolerant quantum computing \cite{NISQ-Preskill} and of a reliable quantum repeater \cite{repeaters} (for latest discovery see \cite{2018Nature-last-repeaters} and references therein) are both different from that of a fault-tolerant \textit{quantum storage} (QS). The memory of the quantum computer need not be stable for a long time, because it is needed only for the time when the gates
of quantum algorithm are done, while the QS needs to be stable for a long period of time. However, operations on the QS are far from being universal \cite{RB2001}, reduced to measurements in two bases (at least in the considered SDI money scheme). In that respect the QS appears to have much easier functionality. The easiness of operations of the QS is more comparable with
that of the single quantum repeater station (at least for the 1st generation quantum repeater \cite{Rep-generations}).
However the $1$st generation single repeater's node needs to achieve the operation
of entanglement swapping of two photons incoming from different origins, which is a totally different task. In the case of the QS states are prepared and need not be send, i.e., QS can be done without the use of photons. This should simplify this task in comparison to the task of achieving quantum Internet. (Note also that a station of the $3$rd generation quantum repeater is close in performance to the small-size universal quantum computer \cite{Rep-generations}).

However, although there is no physical law that bounds from above the time of coherence of a qubit state, achieving a reliable QS appears to be an extremely hard task, because of quantum decoherence, that is usually happening very fast. This is the reason why the very first idea (of QS needed for money schemes) appearing in the theory may become the last one (after quantum computer and quantum Internet) to be realized in practice. 
This may also happen due to the fact, that, in contrast to quantum computing or quantum secure communication, money scheme requires to be widespread implemented in order to be useful.

It is also worth to notice that, recently, the first experimental implementation of quantum money schemes were performed \cite{Bartkiewicz2017,Bozzio2018,Guan2018}. It indicates that real life implementation of quantum money could potentially be achievable using the near future technologies.

\section{Discussion}\label{sec:discusion}
	
In this article we have presented an alternative method of testing
of the original Wiesner banknotes - a Semi-Device Independent quantum money scheme. To our knowledge this is the first attempt to provide the private money scheme unforgeability of which would not fully relay on trusting the mint (source of the banknote) \textit{and} the inner workings of the verifying terminal at the same time.  

Furthermore we provide impossibility result for a fully device independent quantum money. It shows that our Semi-Device Independent approach is a good candidate for the strongest possible money scheme. To clarify, by the strongest we mean here that we maximally reduce trust to inner working of all quantum devices that are used in all stages of money production and verification.

We have proven that the scheme cannot be broken by a forgery
who copies the banknote in a qubit-by-qubit manner in the scenario when
the banknote is returned to the Bank for verification, provided the
banknote was created in a qubit-by-qubit manner (each qubit created
independently). The scheme remains secure in a case of verification by the classical communication at a distance, upon the assumption that the counterfeiter lies about the outputs in the independent manner during the verification procedure. 

We have thereby also made an explicit connection of the money schemes with the idea of a private key which is not classical, but from the other theory (generalized probability theory), exploring thereby directions presented in \cite{Gavinsky2012,qptmoney}. 
 
It is plausible that the proposed scheme 
inherits the security of the underlying, in our case the original semi-device
independent quantum key distribution protocol. Given the full proof
of security of the SDI QKD against a forward signaling adversary, as it is 
the case for the DI QKD \cite{DevInd09} (see \cite{DI-iid} for the latest breakthrough), it may follow that our suitably modified scheme is fully unforgeable. The sufficient modification concerns the communication in verification procedure. The counterfeiter would need to give the answer(s) $(a_A^i,a_F^i)$ after getting inputs $(x_A^i,x_F^i)$,  but \textit{before learning next inputs} $(x_A^{i+1},x_F^{i+1})$. In such a case each possible history-dependent lie can be treated safely as a part of the attack of the device, and hence would not affect the model. The rest of the proof would follow from similar arguments as above with a proper use of the concentration of martingales. It is therefore important to verify if the SDI protocol is fully secure. An intermediate step would be to extend the security proof
the presented SDI money scheme to its variant given in \cite{Wang2014}, prove there to
be secure against collective attacks. 

One might think that we could have used directly the scheme of the device dependent key secure against the quantum adversary \cite{DevInd09}, avoiding thereby unnatural assumption about altering outputs by the terminal in an independent manner during verification procedure. 
It is indeed straightforward to extend the idea presented here for a single Bank's branch with much weaker assumptions. However it needs suitable modifications leading to novel scheme(s), in order to be extended to the case with multiple Bank's branches. This approach therefore results in a  scheme fundamentally different from the original one and its follow-ups like the presented SDI money scheme.

We have compared the SDI money scheme with the protocol of the SDI QKD,
showing that they differ in three ways. Firstly, the money scheme requires a reliable quantum memory, while the SDI QKD does not. Most of quantum money schemes suffer from this problem, i.e., it is not a special property of our scheme  (however, see the recent proposal by A. Kent \cite{1806.05884}). Secondly, in principle, the money scheme does not need the number of runs of the experiment, as producing the key is out of focus, conversely to the goal of the SDI QKD protocol. However, the presented SDI scheme based on qubits leads to the banknotes of the significant quantum memory (as we exemplify in Section \ref{app:mem}), because number of qubits $n$ has to diminish the effect of fluctuations $\eta$. Fortunately, our security proof seems to be straightforwardly adaptable to the SDI schemes based on the SDI QKD protocol with more than two inputs on Bank's side and (if needed) the dimension of the system \cite{RAC1, RAC2}. Considering such an extension would be of high importance for more practical examples. Thirdly, our money scheme needs a moderate error tolerance, roughly speaking just little above the one implying one half of the possible key rate achievable in the corresponding SDI scheme. This, in principle opens an area for the robustness of the money scheme against noise. Given the banknotes are initially prepared at high quality, it can drop, significantly yet without compromising security against forgery, to the value corresponding to about half of the maximal possible key rate of the SDI protocol.

Given a more promising for practical realization variant of this scheme exists, one should consider its robust version, that can be realized in laboratory including all side effects, that may potentially open it for the attacks of hackers. This aspect of
the SDI QKD has been recently studied in \cite{Wang2014, Xu2017, Zhou2017, Chaturvedi2018}.

Another important direction of development would be checking if the
proposed scheme could be treated as an option for a user of the original Wiesner
scheme or its other extensions like Gavinsky's protocol.
The resulting scheme would give higher protection
against malicious money provider, matching the best of two approaches. In the presented scheme, the banknotes (even in case of the honest client Alice) are inevitably lost during their verification. It seems natural then (like it is done by Gavinsky \cite{Gavinsky2012}) to extend our scheme to the case of the \textit{transactions} which we also defer to the future work. 

Finally, in a bit speculative way, we have put forward a hypothesis called the \textit{Quantum Oresme-Copersnicus-Gresham Law}: an analogue of the classical law of the economy also known as Gresham's Law. This law states that the bad money (with lower intrinsic value) drives out the good one (with higher intrinsic value), as 
the latter is less often spend. We have exemplified this law on the basis of different realizations of the SDI money scheme, corresponding to the different values of the threshold leading to the acceptance of money. These speculations need further, more formal, exploration with examples based on more types of currencies, as well as an extension (what appears to be straightforward) to the case of resources within the paradigm of \cite{Multi-Res}.

\begin{acknowledgments}
	The work is supported by National Science Centre grant Sonata Bis 5 no. 2015/18/E/ST2/00327. 
	The authors would like to thank 
	Anubhav Chaturvedi,
	Ryszard P. Kostecki,
	and 
	Marek Winczewski
	for useful comments.
	M.S. thanks Or Sattath for the enlightening discussion about relating results presented as a poster at the 8th International Conference on Quantum Cryptography (QCrypt 2018) Shanghai, China, 27–31 August 2018.
\end{acknowledgments}

\onecolumngrid
\appendix

\section{Preliminary definitions}\label{app:def}

We will start by defining two crucial constants \( M \) and \( P_Q \), which come from \cite{sdi},
\begin{equation}
	M \coloneqq \frac{5 + \sqrt{3}}{4},\quad P_Q \coloneqq \cos^2\left( \frac{\pi}{8} \right).
\end{equation}
It is easy to see that \( P_Q > M/2 \).

Lets us now define notation used in the rest of the paper. By \(y\)'s we denote the inputs used by the Bank in order to create the money, \(x\)'s stand for the questions that the branches verifying Alice and Frederick ask them, \(\tilde{x}\) represents real value that Alice and Frederick input into the devices, and we use \(a\)'s for Alice's and Frederick's outputs. Furthermore, \(i\) in an upper index denotes the \(i\)-th run of the protocol that acts on the \(i\)-th quantum subsystem.
The general attack performed in the qubit-by-qubit manner (See Assumptions \ref{asm:5} and \ref{asm:6}) can be described by a probability measure on the data used in the verification protocol. Part of the data are generated by the Bank (inputs to the verification procedure) while the outputs \( a_A \) and \( a_F \) are generated by the Alice and Frederick according to their choice of the conditional distribution. The total joint distribution of the inputs and outputs reads 
\begin{equation}
	\bigotimes_{i=1}^n P(y^i_{0,B_A}, y^i_{0,B_F}, y^i_{1,B_A}, y^i_{1,B_F}, x^i_A, x^i_F, \tilde{x}^i_A, \tilde{x}^i_F) P\left(a^i_A,a^i_F | y^i_{0,B_A}, y^i_{0,B_F}, y^i_{1,B_A}, y^i_{1,B_F}, x^i_A, x^i_F, \tilde{x}^i_A, \tilde{x}^i_F \right).
\end{equation}
In what follows we will simplify it due to certain assumptions. We know, from the definition of money generating protocol, that if Frederick wants to verify the same banknote as Alice, then \(y\)'s are the same for all branches, so we can omit variables for each branch and write just \( y^i_0 \) and \( y^i_1 \). 
Furthermore, if Bank's branches input appropriate bits to the devices themselves, than we are sure that \( x^i_A = \tilde{x}^i_A\) and \( x^i_F = \tilde{x}^i_F \), obtaining
\begin{equation}
	\bigotimes_{i=1}^n P\left(y^i_0, y^i_1, x^i_A, x^i_F \right) P\left(a^i_A,a^i_F | y^i_0, y^i_1, x^i_A, x^i_F \right).
\end{equation}
\begin{observation}
	Our scheme remains secure if we allow Alice and Frederick to set device inputs, under assumption that they do it in an independent way in each run. Any run-independent cheating strategy of Alice or Frederick based on using inputs \( \tilde{x}^i_A, \tilde{x}^i_F \) different from \( x^i_A, x^i_F \) provided by the Bank can be incorporated into inner working of the untrusted devices and we can also omit it.
\end{observation}
Since we assume that \(y\)'s and \(x\)'s generated by Bank are fully random, what is possible due to Assumption \ref{asm:1}, we can rewrite the above formula as 
\begin{equation}\label{eq:main-dist}
	\mu \sim	\bigotimes_{i=1}^n U \left(y^i_0, y^i_1, x^i_A, x^i_F \right) P\left(a^i_A,a^i_F | y^i_0, y^i_1, x^i_A, x^i_F \right),
\end{equation}
where \(U\), here and in all measures defined later, stands for the uniform distribution over appropriate variables. 

Now we can define the set describing successful forgery, meaning that both Alice and Frederick are accepted using the same banknote.

\begin{equation}
	\mathcal{F} \coloneqq \left\{(y^i_0, y^i_1, x^i_A, a^i_A, x^i_F, a^i_F)_{i=1}^n \in \{ 0,1 \}^{6n} : \left|\left\{ i \in [n] : a^i_A = y^i_{x^i_A} \right\}\right| > \theta n, \left|\left\{ i \in [n] : a^i_F = y^i_{x^i_F} \right\}\right| > \theta n \right\},
\end{equation}
\begin{equation}
	S \coloneqq (y^i_0, y^i_1, x^i_A, a^i_A, x^i_F, a^i_F)_{i=1}^n \in \{ 0,1 \}^{6n}.
\end{equation}
We also define sequences
\begin{equation}
	S^i \coloneqq (y^i_0, y^i_1, x^i_A, a^i_A, x^i_F, a^i_F).
\end{equation} 

Now we can make the following observation that is an easy consequence of the security proof of \cite{sdi}. It is important to notice that we need here Assumptions \ref{asm:1}, \ref{asm:2}, \ref{asm:3}, and \ref{asm:4} since there are also necessary in \cite{sdi} For clarity, we change notation by substituting \(B\) and \(E\) by \(A\) and \(F\), respectively.
\begin{observation}\label{obs:M}
	\begin{equation}
		P_{AF}(a_0) + P_{AF}(a_1) \leq M.
	\end{equation}
\end{observation}
\begin{proof}
	From Eq. (12) of \cite{sdi} and the comment that follows the equation we know that
	\begin{equation}
		P_{AF}(a_0) + P_{AF}(a_1) + P_{AF}(a_0 \oplus a_1) \leq \frac{3}{2} \left( 1 + \frac{1}{\sqrt{3}} \right).
	\end{equation}
	Using Eq. (13) of \cite{sdi},
	\begin{equation}
		P_{AF}(a_0) + P_{AF}(a_1) - 1 \leq P_{AF}(a_0 \oplus a_1),
	\end{equation}	
	we obtain
	\begin{equation}
		P_{AF}(a_0) + P_{AF}(a_1) \leq \frac{1}{2} \left( \frac{3}{2} \left( 1 + \frac{1}{\sqrt{3}} \right) + 1\right) = \frac{5 + \sqrt{3}}{4}.
	\end{equation}
	The right side is equal to \(M\), which completes the proof.
\end{proof}

\section{Main lemmas} \label{app:lem}

We will use numerously the concentration property of a distribution of independently
distributed $n$ random variables on $[0,1]$ due to Hoeffding, of the form
\begin{equation}\label{eq:id-conc}
P( |{\bar X} - \mathbf{E} \bar{X} |\geq \eta ) \leq 2 \mathrm{e}^{- 2 n \eta^2}.
\end{equation}
where ${\bar X} = (1/n) \sum_i X_i$. 

For a bitstring $x$ of length $n$ we will denote
by $|x|_0$ the number of occurrences of symbol $0$ in $x$ (analogously
$|x|_1$ will denote the number of $1$s in $x$). Thus,
\begin{equation}\label{eq:iid-conc}
	P \left( \left| \frac{|x|_0}{n} - \frac{1}{2}\right|\geq \eta \right) \leq 2 \mathrm{e}^{- 2 \eta^2 n},
\end{equation}
where $\eta \geq 0$.  Due to the above concentration, probability
mass function is concentrated on the so-called $\eta$-\textit{typical sequences}, defined as the values of $x$ satisfying $||x|_0/n - 1/2| \leq \eta$.
In other words, for a set 
\begin{equation}
	\mathcal{T}_{n,\eta} \coloneqq \left\{ x : \left| \frac{|x|_0}{n}  - \frac{1}{2} \right| \leq \eta \right\}, 
\end{equation}
there is,
\begin{equation}
	P_{\sim U(x)} ( x \in \mathcal{T}_{n,\eta}) \geq 1 - 2 \mathrm{e}^{- 2 \eta^2 n}
\end{equation}
where the probability is taken from a uniform distribution \(U(x)\) of sequences \( x \coloneqq (x^i)_{i=1}^n \) over $\{0,1\}^n$. In particular, for two sequences $x_A$ and $x_F$ drawn independently at random from $\{0,1\}^n$,
\begin{equation}\label{eq:oplus-typ}
	P (x_A\oplus x_F \in \mathcal{T}_{n,\eta}) \geq 1 - 2 \mathrm{e}^{- 2 \eta^2 n},
\end{equation}
where by $\oplus$ we mean the bit-wise XOR operation on the bits
of $x_A$ and $x_F$. Indeed, for any fixed $x_A$ the distribution
of $x_F\oplus x_A$ is uniform if such was that of $x_F$.  
We can use then the typicality argument and average over $p(x_A)$.

At the expense of small error one can deal only with such as $S$
that have  $\eta$-typical inputs $x_A$ and $x_F$.
Such $S$ will be called $\eta$-\textit{typical}:
\begin{equation}
	S\equiv (y^i_0, y^i_1, x^i_A, a^i_A, x^i_F, a^i_F)_{i=1}^n \mbox{ is called $\eta$-\textit{typical} iff } x_A\oplus x_F \equiv (x_A^i \oplus x_F^i)_{i=1}^n \in \mathcal{T}_{n,\eta}.
\end{equation}
The set of $\eta$-typical $S$ will be denoted as $T(\eta)$.

In what follows we will show that the probability of 
acceptance of a banknote twice, i.e., \( P(\mathcal{F}) \), is equal to the probability
accepting it twice in a different scenario (the XOR scenario). In the latter Alice gets $x_A$ while Frederick is given $x_A\oplus x_F$. In spite
of the fact that it will not be the case in real life, this
transformation of the scenario (and the corresponding probability
measure) will simplify our considerations.

The XOR scenario is obtained from the original one by the following
map on the events $S$:

\begin{equation}
	S = (y^i_0, y^i_1, x^i_A, a^i_A, x^i_F, a^i_F)_{i=1}^n \mapsto^{\pi} S' \coloneqq (y_0^i,y_1^i,x_A^i,a_A^i,x_F^i\oplus x_{A}^i,a_F^i)_{i=1}^n.
\end{equation}

We will refer to the transformed event as the one
having $x_\oplus^i$ on the position where $x_F^i$ is in $S$:
\begin{equation}
	S' \coloneqq (y_0^i,y_1^i,x_A^i,a_A^i,x^i_{\oplus},a_F^i)_{i=1}^n.
\end{equation}

We define the set of all forged $S'$ in a way analogous to the definition of the set $\mathcal{F}$:
\begin{equation}\label{eq:fprim}
\mathcal{F}' \coloneqq \left\{(y^i_0, y^i_1, x^i_A, a^i_A, x^i_{\oplus}, a^i_F)_{i=1}^n \in \{ 0,1 \}^{6n} : \left|\left\{ i \in [n] : a^i_A = y^i_{x^i_A} \right\}\right| > \theta n,
  \left|\left\{ i \in [n] : a^i_F = y^i_{x_\oplus^i \oplus x^i_A} \right\}\right| > \theta n \right\}.
\end{equation}

A new probability measure $\mu'$ defined on the set of 
events $S'$ is defined as
\begin{equation}
	\mu'\sim \bigotimes_{i=1}^n U \left(y^i_0, y^i_1, x^i_A, x^i_\oplus \oplus x^i_A  \right) P\left(a^i_A,a^i_F | y^i_0, y^i_1, x^i_A, x^i_\oplus \oplus x^i_A \right).
\end{equation}
\begin{observation}
	The map $\pi$:
	\begin{enumerate}
		\item is bijective and involutive,
		\item satisfies $S \in \mathcal{F} \Leftrightarrow S' \in \mathcal{F}'$,
		\item satisfies $\mu'(S') = \mu(S)$.
	\end{enumerate}
\end{observation}
\begin{proof}
	The bijectivity follows directly from the fact that
	$(x_A,x_F)$ is bijectively mapped to $(x_A,x_A\oplus x_F)$.
	The first input is preserved, while the second one can
	be reconstructed uniquely by XORing inputs. It is also easy to see that \( \pi \) is an involution, since $(x_A,x_A\oplus x_F)$ is mapped back to $(x_A,x_F)$.
	
	We show now the Property 2. Let $S \in \mathcal{F}$.
	This happens if and only if 
	\begin{equation}
		\left|\left\{ i \in [n] : a^i_A = y^i_{x^i_A} \right\}\right| > \theta n \text{ and } \left|\left\{ i \in [n] : a^i_F = y^i_{x^i_F} \right\}\right| > \theta n.
	\end{equation}
	The event $S'$ equals 
	$(y_0^i,y_1^i,x_A^i,a_A^i,x_\oplus^i,a_F^i)_{i=1}^n$.
	By definition of \(\mathcal{F}' \), we have that \( S' \in \mathcal{F}' \) if and only if
	\begin{equation}
	\left|\left\{ i \in [n] : a^i_A = y^i_{x^i_A} \right\}\right| > \theta n \text{ and } \left|\left\{ i \in [n] : a^i_F = y^i_{x_\oplus^i \oplus x^i_A} \right\}\right| > \theta n.
	\end{equation}
	Since the left conditions are identical, we only have to prove equality on the right conditions. By definition of a map \( \pi^{-1} \), we obtain
	\begin{equation}
	\left|\left\{ i \in [n] : a^i_F = y^i_{x_\oplus^i \oplus x^i_A} \right\}\right| = \left|\left\{ i \in [n] : a^i_F = y^i_{( x^i_A \oplus x^i_F ) \oplus x^i_A} \right\}\right| = \left|\left\{ i \in [n] : a^i_F = y^i_{ x^i_F} \right\}\right|,
	\end{equation}
	what proves an appropriate equality and implies that $S \in \mathcal{F} \Leftrightarrow S' \in \mathcal{F}'$.

	Finally, we argue that the Property 3 also holds.
	Let us fix arbitrary $S$. Hence, $x_{\oplus}^i = x_{A}^i\oplus x_F^i$ in definition of $S'$, and
	\begin{equation}
	\begin{split}
	\mu'(S') &= \bigotimes_{i=1}^n U \left(y^i_0, y^i_1, x^i_A, x^i_\oplus \oplus x^i_A  \right) P\left(a^i_A,a^i_F | y^i_0, y^i_1, x^i_A, x^i_\oplus \oplus x^i_A \right) \\
	&\stackrel{\pi^{-1}}{=}	\bigotimes_{i=1}^n U \left(y^i_0, y^i_1, x^i_A, (x^i_F\oplus x_A^i) \oplus x^i_A \right) P\left(a^i_A,a^i_F | y^i_0, y^i_1, x^i_A, (x^i_F\oplus x_A^i) \oplus x^i_A \right) \\
	&=	\bigotimes_{i=1}^n U \left(y^i_0, y^i_1, x^i_A, x^i_F \right) P\left(a^i_A,a^i_F | y^i_0, y^i_1, x^i_A, x^i_F \right) = \mu(S),
	\end{split}
	\end{equation}
	where \( \pi^{-1} \) above denotes that equality follows from the properties of the inverse of map \( \pi \), which due to involution property is equal to \( \pi \) .
\end{proof}

Alice, as before, gets bit \( x^i_A \), but Frederick obtains XOR of bits \( x^i_A \) and \( x^i_F \). Despite this, ``original'' box, due to ``wirings'',  receives \( x^i_A \) and \( x^i_F \).
We have then an important corollary, that we can focus now
on the XOR scenario because the probability of
forgery in the latter equals to the probability of forgery
in the former. 
\begin{corollary} \label{cor:accacc'}
	\begin{equation}\label{eq:accacc'}
	P_{\sim \mu}(\mathcal{F}) = P_{\sim\mu'}(\mathcal{F}').
	\end{equation}
\end{corollary}

One can focus on the typical sequences \(S\), i.e., those for which \( \textbf{x}_\oplus \in \mathcal{T}(\eta) \), at the expense of exponentially small inaccuracy in estimating the probability of forgery due to measure $\mu'$.

\begin{lemma} \label{lem:epstyp}
	\begin{equation}\label{eq:epstyp}
	P(\mathcal{F}') \leq P(\mathcal{F}' \cap T(\eta)) + \epsilon(\eta),
	\end{equation}
	with $\epsilon(\eta) = 2 \exp (- 2 \eta^2 n)$.
\end{lemma}
\begin{proof}
	With a little abuse of notation we will mean by 
	$(y_0,y_1,x_A,a_A,x_\oplus,a_F)$ the properly ordered sequence of tuples 	$(y_0^i,y_1^i,x_A^i,a_A^i,x_\oplus^i,a_F^i)_{i=1}^n$,
	where $y_0 = (y_0^i)_{i=1}^n$, and by analogy the same for other symbols.
	
	We will show first a sequence of (in)equalities:
	\begin{equation}
	\begin{split}
	P(\mathcal{F}') =& \sum_{S'\in \mathcal{F}' } P_{\sim \mu'}(S')\\
	=&	\sum_{ S'\in \mathcal{F}' \cap T(\eta) } P_{\sim \mu'}(S') + \sum_{ S'\in \mathcal{F}'\setminus T(\eta) } P_{\sim \mu'}(S') \\
	=& \sum_{x_A,x_\oplus,y_0,y_1 \subset S' : S' \in \mathcal{F}' \cap T(\eta)} u(x_A) u(x_\oplus) u(y_0)u(y_1) P(a_A,a_F|y_0,y_1,x_A,x_\oplus\oplus x_A)  \\
	+&	\sum_{x_A,x_\oplus,y_0,y_1 \subset S' : S' \in \mathcal{F}' \setminus T(\eta)} u(x_A) u(x_\oplus) u(y_0)u(y_1) P(a_A,a_F|y_0,y_1,x_A,x_\oplus\oplus x_A)  \\
	\leq&	\sum_{x_A,x_\oplus,y_0,y_1 \subset S' : S' \in \mathcal{F}'\cap T(\eta)} u(x_A) u(x_\oplus) u(y_0)u(y_1) P(a_A,a_F|y_0,y_1,x_A,x_\oplus\oplus x_A)  \\
	+&	\sum_{x_A,x_\oplus,y_0,y_1 : (x_\oplus)\in \mathcal{T}(\eta)  } u(x_A) u(x_\oplus) u(y_0)u(y_1)   \\
	\leq&	P(\mathcal{F}'\cap T(\eta)) + \epsilon(\eta).
	\end{split}
	\end{equation}
	
	We have used the fact that the distribution of $x_\oplus$ is the same (uniform) irrespectively of a particular
	attack. This is because the distributions of $x_A$ and $x_F$ with respect to the measure $\mu$  are uniform and independent from the attack, as being prior to the attack, while $x_\oplus$ has distribution of $x_A\oplus x_F$ according to the definition of a measure $\mu'$. In the last inequality we have used typicality argument from Eq. (\ref{eq:oplus-typ}).
\end{proof}

Due to random nature of variable \(x_\oplus\),
\begin{equation}\label{eq:accgivenxor}
P(\mathcal{F}' \cap T(\eta)) = \sum_{x_\oplus\in T(\eta)} P(x_\oplus) P(\mathcal{F}'|x_\oplus).
\end{equation}

The advantage of the measure $\mu'$ is that we can easily
divide the set of each run $i$ according to the values of 
the $x_{\oplus}^i$. Technical as it sounds, it will simplify
the argument. In the runs where $x_{\oplus}^i=0$, the best strategy
achieves quantum value $P_Q$ for both Alice and Frederick. However,
for $x_{\oplus}=1$ they are in a position of Alice guessing 
the \textit{opposite} bit to the one which Frederick is at the same
time in this run $i$ to guess. Hence, they are limited 
as it is shown in the original paper by Paw\l{}owski and Brunner 
\cite{sdi}.

From now on, we will fix $ \bm{x}_\oplus \coloneqq (x^i_\oplus)_{i=1}^n$ and prove
a common bound on guessing for all of its typical values.
We can then define new conditional measure \( \mu'' \) that depends on $ \bm{x}_\oplus $ as
\begin{equation} \label{eq:mu''}
\mu'' \coloneqq \mu_{|\tilde{\bm{x}}_\oplus} = \delta(\tilde{\bm{x}}_\oplus, \bm{x}_\oplus) \bigotimes_{i=1}^n u(x^i_A) u(y^i_0) u (y^i_1) P^i(a^i_F,a^i_A | x^i_A, x^i_\oplus \oplus x^i_A, y^i_0, y^i_1).
\end{equation}

We will now show that, on average, the 
forgeries Alice and Frederick have total number of correctly guessed
bits of $y_0,y_1$ bit-strings bounded from above by certain value.
Let us first define the set of indexes 
\begin{equation} \label{eq:D}
D(\bm{x}_\oplus) \coloneqq \{ i \in [n] : x^i_\oplus = 0 \},
\end{equation}
and its complement \( \bar{D}(\bm{x}_\oplus) \). It is important to notice that, for runs in the set \( D(\bm{x}_\oplus) \), Alice and Frederick will be asked about the same Bank's bit and in the case of \( \bar{D}(\bm{x}_\oplus) \) they will have to guess two different bits of the Bank.

Then we can consider four types or random variables defined on \( \Omega_i \), each
depending on the value of the $\bm{x}_\oplus$. It is important to notice that these variables describe the probabilities of guessing appropriate bits in \(i\)th run, by Alice and Frederick respectively and furthermore have strong connection with the definition of \( \mathcal{F} \).
\begin{equation}
X_{i,A}^{\bm{x}_\oplus} \coloneqq 
\begin{cases}
\delta\left(a^i_A,y^i_{x_A^i}\right) &: i \in D(\bm{x}_\oplus) \\
0 &: i \in {\bar D}(\bm{x}_{\oplus})
\end{cases},
\qquad
X_{i,F}^{\bm{x}_\oplus} \coloneqq 
\begin{cases}
\delta\left(a^i_F,y^i_{x_A^i}\right) &: i \in D(\bm{x}_\oplus) \\
0 &: i \in {\bar D}(\bm{x}_{\oplus})
\end{cases},
\end{equation}
\begin{equation}
Y_{i,A}^{\bm{x}_\oplus} \coloneqq 
\begin{cases}
\delta\left(a^i_A,y^i_{x_A^i}\right) &: i \in \bar{D}(\bm{x}_\oplus) \\
0 &: i \in D(\bm{x}_{\oplus})
\end{cases},
\qquad
Y_{i,F}^{\bm{x}_\oplus} \coloneqq 
\begin{cases}
\delta\left(a^i_F,y^i_{x_A^i \oplus 1}\right) &: i \in \bar{D}(\bm{x}_\oplus) \\
0 &: i \in D(\bm{x}_{\oplus})
\end{cases},
\end{equation}
where the sample space is defined as
\begin{equation}
\Omega^i_{\bm{x}_\oplus} \coloneqq \{ S'^i : x^i_\oplus(S'^i) = x_\oplus^i \},
\end{equation}
 while \( x^i_\oplus(S'^i) \) denotes taking a variable with a label \( x^i_\oplus \) from the sequence \(S'^i\).
We will also define sums
\begin{equation}
X_{i}^{\bm{x}_\oplus} \coloneqq X_{i,A}^{\bm{x}_\oplus} + X_{i,F}^{\bm{x}_\oplus}, \quad Y_{i}^{\bm{x}_\oplus} \coloneqq Y_{i,A}^{\bm{x}_\oplus} + Y_{i,F}^{\bm{x}_\oplus},
\end{equation}
and the random variables built from
$X_i$ and $Y_i$, that is their sums:
\begin{align}
{\bar X}^A_{\bm{x}_\oplus} \coloneqq \sum_{i \in [n]} X_{i,A}^{\bm{x}_\oplus}, \quad {\bar X}^F_{\bm{x}_\oplus} \coloneqq \sum_{i \in [n]} X_{i,F}^{\bm{x}_\oplus}, \quad {\bar X}_{\bm{x}_\oplus} \coloneqq \sum_{i \in [n]} X_i^{\bm{x}_\oplus}, \\
{\bar Y}^A_{\bm{x}_\oplus} \coloneqq \sum_{i \in [n]} Y_{i,A}^{\bm{x}_\oplus}, \quad {\bar Y}^F_{\bm{x}_\oplus} \coloneqq \sum_{i \in [n]} Y_{i,F}^{\bm{x}_\oplus}, \quad {\bar Y}_{\bm{x}_\oplus} \coloneqq \sum_{i \in [n]} Y_i^{\bm{x}_\oplus}.
\end{align}
From Eq. (\ref{eq:mu''}) we know that the measure \( \mu'' \) is a product of measures 
\begin{equation} \label{eq:mu''_i}
\mu''_i \coloneqq u(x^i_A) u(y^i_0) u (y^i_1) P^i(a^i_F,a^i_A | x^i_A, x^i_\oplus \oplus x^i_A, y^i_0, y^i_1)
\end{equation}
what implies that \( {\bar X}^A_{\bm{x}_\oplus}, {\bar X}^F_{\bm{x}_\oplus}, {\bar X}_{\bm{x}_\oplus}, {\bar Y}^A_{\bm{x}_\oplus}, {\bar Y}^F_{\bm{x}_\oplus}, {\bar Y}_{\bm{x}_\oplus} \) are described by Poisson distribution.

The respective averages over measure
$\mu''$ read
\begin{align}
\mathbf{E} {\bar X}_{\bm{x}_\oplus} = \sum_{i \in [n]} \sum_{(x_A^i,y_0^i,y_1^i)= (\hat{x}_A^i,\hat{y}_0^i,\hat{y}_1^i)\in \{0,1\}^3} \frac{1}{8} \sum_{(a_A^i,a_F^i)\in\{0,1\}^2} P_i(a_A^i,a_F^i|\hat{x}_A^i,\hat{x}_A^i\oplus x_\oplus^i, \hat{y_0}^i,\hat{y_1}^i) X^{\bm{x}_\oplus}_i(S'^i), \\
\mathbf{E} {\bar Y}_{\bm{x}_\oplus} = \sum_{i \in [n]} \sum_{(x_A^i,y_0^i,y_1^i)= 
	(\hat{x}_A^i,\hat{y}_0^i,\hat{y}_1^i)\in \{0,1\}^3} \frac{1}{8}  \sum_{(a_A^i,a_F^i)\in\{0,1\}^2}
P_i(a_A^i,a_F^i|\hat{x}_A^i,\hat{x}_A^i\oplus x_\oplus^i, \hat{y_0}^i,\hat{y_1}^i) Y^{\bm{x}_\oplus}_i(S'^i).
\end{align}

Although the above averages are defined on the whole range $[n]$, they depend only
on their respective subsets:
\begin{equation}
\mathbf{E} {\bar X}_{\bm{x}_\oplus} = \sum_{i \in D(\bm{x}_\oplus)}  \sum_{(x_A^i,y_0^i,y_1^i)= 	(\hat{x}_A^i,\hat{y}_0^i,\hat{y}_1^i)\in \{0,1\}^3} \frac{1}{8} \sum_{(a_A^i,a_F^i)\in\{0,1\}^2}P_i(a_A^i,a_F^i|\hat{x}_A^i,\hat{x}_A^i, \hat{y_0}^i,\hat{y_1}^i)
 \left[\delta\left(a^i_A,\hat{y}^i_{x_A^i}\right) + \delta\left(a^i_F,\hat{y}^i_{x_A^i}\right)\right],
\end{equation}
\begin{equation}
\begin{split}
\mathbf{E} {\bar Y}_{\bm{x}_\oplus} = \sum_{i \in \bar{D}(\bm{x}_\oplus)} \sum_{(x_A^i,y_0^i,y_1^i)= 	(\hat{x}_A^i,\hat{y}_0^i,\hat{y}_1^i)\in \{0,1\}^3} \frac{1}{8} \sum_{(a_A^i,a_F^i)\in\{0,1\}^2} P_i(a_A^i,a_F^i|\hat{x}_A^i,\hat{x}_A^i\oplus 1, \hat{y_0}^i,\hat{y_1}^i)\\
\times \left[\delta\left(a^i_A,\hat{y}^i_{x_A^i}\right) + \delta\left(a^i_F,\hat{y}^i_{1\oplus x_A^i}\right)\right].
\end{split}
\end{equation}

We will prove now, that the above averages are bounded if the
attack is done under assumption that the adversaries are quantum and they attack in a qubit-by qubit manner (see Assumptions \ref{asm:3}--\ref{asm:6}).

\begin{lemma}\label{lem:u-bound}
	\begin{equation}\label{eq:main-up}
	\mathbf{E} {\bar X}_{\bm{x}_\oplus} + \mathbf{E} {\bar Y}_{\bm{x}_\oplus}  \leq 2 P_Q \left|D(\textbf{x}_\oplus)\right| + M \left|\bar{D}(\textbf{x}_\oplus)\right|.
	\end{equation}
\end{lemma}
\begin{proof}
	We will separately prove that the first term of LHS is bounded 
	by the first term of RHS, and later that the second terms bound
	each other, respectively. The best quantum strategy for a single person (say, Alice) in a single run of experiment is upper bounded by $P_Q$, while the other
	party is asked to guess \textit{the same bit} as Alice was asked, 
	so can copy her answer. We have

	\begin{equation}
	\begin{split}
	\mathbf{E} {\bar X}_{\bm{x}_\oplus} =& \sum_{i \in D(\bm{x}_\oplus)} \sum_{(x_A^i,y_0^i,y_1^i)= 		(\hat{x}_A^i,\hat{y}_0^i,\hat{y}_1^i)\in \{0,1\}^3} \frac{1}{8} \sum_{(a_A^i,a_F^i)\in\{0,1\}^2}	P_i(a_A^i,a_F^i|\hat{x}_A^i,\hat{x}_A^i, \hat{y_0}^i,\hat{y_1}^i) \left[\delta\left(a^i_A,\hat{y}^i_{x_A^i}\right) + \delta\left(a^i_F,\hat{y}^i_{x_A^i}\right)\right] \\
	=& \sum_{i \in D(\bm{x}_\oplus)} \sum_{(x_A^i,y_0^i,y_1^i)= 	(\hat{x}_A^i,\hat{y}_0^i,\hat{y}_1^i)\in \{0,1\}^3} \frac{1}{8} \sum_{(a_A^i,a_F^i)\in\{0,1\}^2} P_i(a_A^i,a_F^i|\hat{x}_A^i,\hat{x}_A^i,	\hat{y_0}^i,\hat{y_1}^i) \delta\left(a^i_A,\hat{y}^i_{x_A^i}\right)\\
	+&	\sum_{i \in D(\bm{x}_\oplus)}  \sum_{(x_A^i,y_0^i,y_1^i)= (\hat{x}_A^i,\hat{y}_0^i,\hat{y}_1^i)\in \{0,1\}^3} \frac{1}{8} \sum_{(a_A^i,a_F^i)\in\{0,1\}^2} P_i(a_A^i,a_F^i|\hat{x}_A^i,\hat{x}_A^i, \hat{y_0}^i,\hat{y_1}^i) \delta\left(a^i_F,\hat{y}^i_{x_A^i}\right)\\
	=&	\sum_{i \in D(\bm{x}_\oplus)} \sum_{(x_A^i,y_0^i,y_1^i)= (\hat{x}_A^i,\hat{y}_0^i,\hat{y}_1^i)\in \{0,1\}^3} \frac{1}{8} \sum_{(a_A^i)\in\{0,1\}} P_i(a_A^i|\hat{x}_A^i,\hat{x}_A^i, \hat{y_0}^i,\hat{y_1}^i) \delta \left(a^i_A,\hat{y}^i_{x_A^i}\right)\\
	+& \sum_{i \in D(\bm{x}_\oplus)} \sum_{(x_A^i,y_0^i,y_1^i)= 	(\hat{x}_A^i,\hat{y}_0^i,\hat{y}_1^i)\in \{0,1\}^3} \frac{1}{8}  \sum_{(a_F^i)\in\{0,1\}}P_i(a_F^i|\hat{x}_A^i,\hat{x}_A^i,\hat{y_0}^i,\hat{y_1}^i) \delta\left(a^i_F,\hat{y}^i_{x_A^i}\right)\\
	=& \sum_{i \in D(\bm{x}_\oplus) } \mathbf{E} X_{i,\bm{x}_\oplus}^A + \mathbf{E} X_{i,\bm{x}_\oplus}^F.
	\end{split}
	\end{equation}
	
	Now, for each $i \in D$, there is:
	\begin{equation}
	\mathbf{E} X_{i,\bm{x}_\oplus}^A = \frac{1}{2} \left(P_A(y_0|x_A=0) + P_A(y_1|x_A=1)\right) \leq P_Q,
	\end{equation}
	where we here used Eq. (7) of \cite{sdi} (note the change of notation: our $x, y_i, a_A$ correspond to \( y, a_i, b \) there, respectively).
	Analogously, for $i \in D$,
	\begin{equation}
	\mathbf{E} X_{i,\bm{x}_\oplus}^F = \frac{1}{2} \left(P_F(y_0|x_F=0) + P_F(y_1|x_F=1)\right) \leq P_Q.
	\end{equation}
	
	Summing the above inequalities over $i\ \in D$ we obtain $\mathbf{E} {\bar X}_{\bm{x}_\oplus} \leq 2 P_Q |D(\bm{x}_\oplus)|$. 
	
	More elaborative is relating the second terms of (\ref{eq:main-up}).
	We begin analogously:
	
	\begin{equation}
	\begin{split}
	\mathbf{E}{\bar Y}_{\bm{x}_\oplus} =&  \sum_{i \in \bar{D}(\bm{x}_\oplus)} \sum_{(x_A^i,y_0^i,y_1^i)= 
		(\hat{x}_A^i,\hat{y}_0^i,\hat{y}_1^i)\in \{0,1\}^3}\frac{1}{8}  \sum_{(a_A^i,a_F^i)\in\{0,1\}^2}
	P_i(a_A^i,a_F^i|\hat{x}_A^i,\hat{x}_A^i\oplus 1,
	\hat{y_0}^i,\hat{y_1}^i)\\
	&\qquad\qquad\qquad\qquad\qquad\qquad\qquad\qquad\qquad\qquad\qquad \times \left[\delta\left(a^i_A,\hat{y}^i_{x_A^i}\right) + \delta\left(a^i_F,\hat{y}^i_{1\oplus x_A^i}\right)\right]\\
	=& \sum_{i \in \bar{D}(\bm{x}_\oplus)} \sum_{(x_A^i,y_0^i,y_1^i)= 
		(\hat{x}_A^i,\hat{y}_0^i,\hat{y}_1^i)\in \{0,1\}^3}\frac{1}{8}  \sum_{(a_A^i,a_F^i)\in\{0,1\}^2}
	P_i(a_A^i,a_F^i|\hat{x}_A^i,\hat{x}_A^i\oplus 1,
	\hat{y_0}^i,\hat{y_1}^i)\cdot \delta\left(a^i_A,\hat{y}^i_{x_A^i}\right)\\
	+& \sum_{i \in \bar{D}(\bm{x}_\oplus)} \sum_{(x_A^i,y_0^i,y_1^i)= 
		(\hat{x}_A^i,\hat{y}_0^i,\hat{y}_1^i)\in \{0,1\}^3}\frac{1}{8}  \sum_{(a_A^i,a_F^i)\in\{0,1\}^2}
	P_i(a_A^i,a_F^i|\hat{x}_A^i,\hat{x}_A^i\oplus 1,
	\hat{y_0}^i,\hat{y_1}^i)\cdot \delta\left(a^i_F,\hat{y}^i_{1\oplus x_A^i}\right)\\
	=& \sum_{i\in \bar{D}(\bm{x}_\oplus)} \mathbf{E}Y_{i,\bm{x}_\oplus}^A + \mathbf{E}Y_{i,\bm{x}_\oplus}^F.
	\end{split}
	\end{equation}
	
	For each $i \in \bar{D}(\bm{x}_\oplus)$,
	\begin{equation}
	\begin{split}
	\mathbf{E}Y_{i,\bm{x}_\oplus}^A + \mathbf{E}Y_{i,\bm{x}_\oplus}^F =& \frac{1}{2}\left[P_{A}(y_0^i|x_A^i=0) + P_{F}(y_1^i|x_A^i=0)+P_{A}(y_1^i|x_A^i=1) + P_{F}(y_0^i|x_A^i=1)\right]\\
	\leq& \frac{1}{2}\left[P_{AF}(y_0^i|x_A^i=0) + P_{AF}(y_1^i|x_A^i=0)+P_{AF}(y_1^i|x_A^i=1) + P_{AF}(y_0^i|x_A^i=1)\right]\\
	=& \frac{1}{2}\left[(P_{AF}(y_0^i|x_A^i=0) + P_{AF}(y_0^i|x_A^i=1) \right] + \frac{1}{2} \left[P_{AF}(y_1^i|x_A^i=0) + P_{AF}(y_1^i|x_A^i=1) \right]\\
	=& P_{AF}(y_0^i) + P_{AF}(y_1^i) \leq M,
	\end{split}
	\end{equation}
	where in the pre-last inequality we have used the fact,
	that Alice and Frederic can collaborate.This can only increase
	the average probability of guessing. In the last inequality
	we have rephrased the results of \cite{sdi} as in Observation \ref{obs:M}.
	Summing up over $i\in \bar{D}(\bm{x}_\oplus)$, we obtain
	\begin{equation}
	\sum_{i\in \bar{D}(\bm{x}_\oplus)} \mathbf{E}Y^A_{i,\bm{x}_\oplus} + \mathbf{E}Y^F_{i,\bm{x}_\oplus} \leq M \left|\bar{D}(\bm{x}_\oplus)\right|, 
	\end{equation}
	as it was claimed.
\end{proof}

We have shown above that the average number of guessed
bits has an upper bound. We are going now to argue about the concentration properties of the random variables which are involved in the above process.

We assume here that $\bm{x}_\oplus$ is fixed, and results
in the well defined sets $D(\bm{x}_\oplus)$ and $\bar{D}(\bm{x}_\oplus)$. 
For brevity, we will sometimes omit $\bm{x}_\oplus$ from the notation of $D$. We will also define subsequences \(S'_0\) and \(S'_1\) of particular realization of strategy \(S'\),
\begin{equation}
S'_0 \coloneqq (S'^i)_{ \{ i : x_\oplus^i = 0 \} },\qquad
S'_1 \coloneqq (S'^i)_{ \{ i : x_\oplus^i = 1 \} }.
\end{equation}

In the spirit of the above technical lemma, we will
consider four random variables, each reporting the distance between the theoretical average value of a number of guessed inputs and the observed number of guessed inputs for the respective dishonest party on the respective set ($D(\bm{x}_\oplus)$ or $\bar{D}(\bm{x}_\oplus)$),
\begin{align}
\begin{split}
c_1 &= \begin{cases}
1 : \left|\mathbf{E}\bar{X}^{A}_{\bm{x}_\oplus} - \bar{X}^{A}_{\bm{x}_\oplus}\right| \leq \eta|D(\bm{x}_\oplus)|\\
0 : \text{else}
\end{cases},\qquad
c_2 = \begin{cases}
1 : \left|\mathbf{E}\bar{Y}^{A}_{\bm{x}_\oplus} - \bar{Y}^{A}_{\bm{x}_\oplus}\right| \leq \eta|\bar{D(\bm{x}_\oplus)}|\\
0 : \text{else}
\end{cases},\\
c_3 &= \begin{cases}
1 : \left|\mathbf{E}\bar{X}^{F}_{\bm{x}_\oplus} - \bar{X}^{F}_{\bm{x}_\oplus}\right| \leq \eta|D(\bm{x}_\oplus)|\\
0 : \text{else}
\end{cases},\qquad
c_4 = \begin{cases}
1 : \left|\mathbf{E}\bar{Y}^{F}_{\bm{x}_\oplus} - \bar{Y}^{F}_{\bm{x}_\oplus}\right| \leq \eta|\bar{D}(\bm{x}_\oplus)|\\
0 : \text{else}
\end{cases}.
\end{split}\end{align}

Due to the concentration property given in Eq. (\ref{eq:id-conc}) on the total measure $\mu''$, that is product of measures \( \mu''_i \) (see Eq. (\ref{eq:mu''_i})),
base on which the joined distribution of  $X_1,...,X_4$ random variables defined above, we can bound probability of $0$'s as follows:
\begin{align}
p(c_1 = 0) \leq& 2\mathrm{e}^{-2\eta^2 |D(\bm{x}_\oplus)|} \eqqcolon \epsilon_1(\eta),\qquad
p(c_2 = 0) \leq 2\mathrm{e}^{-2\eta^2 |\bar{D}(\bm{x}_\oplus)|} \eqqcolon \epsilon_2(\eta),\\
p(c_3 = 0) \leq& 2\mathrm{e}^{-2\eta^2 |D(\bm{x}_\oplus)|} \eqqcolon \epsilon_3(\eta),\qquad
p(c_4 = 0) \leq 2\mathrm{e}^{-2\eta^2 |\bar{D}(\bm{x}_\oplus)|} \eqqcolon \epsilon_4(\eta),
\end{align}
where the probability of the above measures is taken 
over the measure $\mu''$.

Now, thanks to the union bound, we obtain:
\begin{corollary}\label{cor:deviation}
	 For any $\bm{x}_\oplus$, there is
	\begin{equation}
	\begin{split}
	P_{\sim \mu''} & \left( \left|\mathbf{E}\bar{X}^{A}_{\bm{x}_\oplus} - \bar{X}^{A}_{\bm{x}_\oplus}\right|\leq \eta\left|D(\bm{x}_\oplus)\right| \right.
	\wedge \left|\mathbf{E}\bar{Y}^{A}_{\bm{x}_\oplus} - \bar{Y}^{A}_{\bm{x}_\oplus}\right| \leq \eta\left|\bar{D}(\bm{x}_\oplus)\right|\\
	&\wedge \left|\mathbf{E}\bar{X}^{F}_{\bm{x}_\oplus} - \bar{X}^{F}_{\bm{x}_\oplus}\right| \leq \eta\left|D(\bm{x}_\oplus)\right|
	 \left. \wedge \left|\mathbf{E}\bar{Y}^{F}_{\bm{x}_\oplus} - \bar{Y}^{F}_{\bm{x}_\oplus}\right| \leq \eta\left|\bar{D}(\bm{x}_\oplus)\right| \right)\\
	\geq& 1 - \sum_i \epsilon_i(\eta).
	\end{split}
	\end{equation}
	\label{cor:main}
\end{corollary}

The above rather technical results are summarized bellow
in the upper bound on the total number of guesses. Namely, 
we will show that for a fixed $\bm{x}_\oplus$, and an attack 
defining the measure $\mu''=\mu'_{|\bm{x}_\oplus}$, the 
random variable of total number of guesses defined as a
function of $S'(\bm{x}_\oplus)$ sampled from $\mu''$
is bounded from above with high probability, as
it is close to the sum of averages that are bounded. Indeed, let us define the random variable of total number of guesses,
\begin{equation}
\bar{Z}_{\bm{x}_\oplus} \coloneqq  \bar{X}^{A}_{\bm{x}_\oplus} + \bar{X}^{F}_{\bm{x}_\oplus} + \bar{Y}^{A}_{\bm{x}_\oplus} + \bar{Y}^{F}_{\bm{x}_\oplus}.
\end{equation}
We additionally define two other useful variables,
\begin{equation}
\bar{Z}^A_{\bm{x}_\oplus} \coloneqq  \bar{X}^{A}_{\bm{x}_\oplus}  + \bar{Y}^{A}_{\bm{x}_\oplus},\qquad
\bar{Z}^F_{\bm{x}_\oplus} \coloneqq   \bar{X}^{F}_{\bm{x}_\oplus} +  \bar{Y}^{F}_{\bm{x}_\oplus}.
\end{equation} 
\begin{lemma} \label{lem:small-event}
	For any fixed $\bm{x}_\oplus \in T(\eta)$,
	\begin{equation}
	P_{\sim \mu''}(\bar{Z}_{\bm{x}_\oplus} > B) \leq \sum_{i=1}^4 \epsilon_i(\eta),
	\end{equation}	
	where
	\begin{equation}
	B \coloneqq 2 P_Q\left(\frac{1}{2} + \eta\right)n + M \left(\frac{1}{2}+\eta\right)n +2\eta n.
	\end{equation}
\end{lemma}
\begin{proof}
	From the Corollary \ref{cor:main}, omitting the modulus, we obtain a sequence of inequalities
	\begin{equation}
	\begin{split}
	1 - \sum_{i=1}^4\epsilon_i(\eta) \leq &  P_{\sim \mu''}\left( \left|\mathbf{E}\bar{X}^{A}_{\bm{x}_\oplus} - \bar{X}^{A}_{\bm{x}_\oplus}\right| \leq \eta\left|D(\bm{x}_\oplus)\right|\right.
	\wedge \left|\mathbf{E}\bar{Y}^{A}_{\bm{x}_\oplus} - \bar{Y}^{A}_{\bm{x}_\oplus}\right| \leq \eta\left|\bar{D}(\bm{x}_\oplus)\right|\\
	&\qquad\wedge \left|\mathbf{E}\bar{X}^{F}_{\bm{x}_\oplus} - \bar{X}^{F}_{\bm{x}_\oplus}\right| \leq \eta\left|D(\bm{x}_\oplus)\right|
	\wedge \left. \left|\mathbf{E}\bar{Y}^{F}_{\bm{x}_\oplus} - \bar{Y}^{F}_{\bm{x}_\oplus}\right| \leq \eta\left|\bar{D}(\bm{x}_\oplus)\right| \right)\\
	\leq &  P_{\sim \mu''}\left( 
	\mathbf{E}\bar{X}^{A}_{\bm{x}_\oplus} \leq  \bar{X}^{A}_{\bm{x}_\oplus} + \eta \left|D(\bm{x}_\oplus)\right|\right.
	\wedge \mathbf{E}\bar{Y}^{A}_{\bm{x}_\oplus} \leq \bar{Y}^{A}_{\bm{x}_\oplus} + \eta\left|\bar{D}(\bm{x}_\oplus)\right|\\
	&\qquad\wedge \mathbf{E}\bar{X}^{F}_{\bm{x}_\oplus} \leq \bar{X}^{F}_{\bm{x}_\oplus} + \eta\left|D(\bm{x}_\oplus)\right|
	\left. \wedge  \mathbf{E}\bar{Y}^{F}_{\bm{x}_\oplus} \leq \bar{Y}^{F}_{\bm{x}_\oplus} + \eta\left|\bar{D}(\bm{x}_\oplus)\right| 
	\right)\\
	\leq&  P_{\sim \mu''}\left( \bar{X}^{A}_{\bm{x}_\oplus} + \bar{Y}^{A}_{\bm{x}_\oplus} +
	\bar{X}^{F}_{\bm{x}_\oplus} + \bar{Y}^{F}_{\bm{x}_\oplus} \leq \mathbf{E}\bar{X}^{A}_{\bm{x}_\oplus} + \mathbf{E}\bar{Y}^{A}_{\bm{x}_\oplus}
	+ \mathbf{E}\bar{X}^{F}_{\bm{x}_\oplus} + \mathbf{E}\bar{Y}^{F}_{\bm{x}_\oplus} +
	2\eta n \right)\\
	\leq& P_{\sim \mu''}\left(\bar{X}^{A}_{\bm{x}_\oplus} + \bar{Y}^{A}_{\bm{x}_\oplus} +
	\bar{X}^{F}_{\bm{x}_\oplus} + \bar{Y}^{F}_{\bm{x}_\oplus} \leq 2 P_Q|D(\bm{x}_\oplus)| + M |\bar{D}(\bm{x}_\oplus)| +2\eta n \right)\\
	=& P_{\sim \mu''}\left(\bar{Z}_{\bm{x}_\oplus} \leq 2 P_Q|D(\bm{x}_\oplus)| + M |\bar{D}(\bm{x}_\oplus)| +2\eta n \right)\\
	\leq& P_{\sim \mu''}(\bar{Z}_{\bm{x}_\oplus} \leq B).
	\end{split}
	\end{equation}
	In the second inequality we have used Lemma \ref{lem:u-bound}. In the next one we have used definition
	of $\bar{Z}_{\bm{x}_\oplus}$, and then we have used the typicality 
	of $\bm{x}_\oplus$, which implies upper bounds on the 
	power of sets $D(\bm{x}_\oplus)$ and $\bar{D}(\bm{x}_\oplus)$.
	This implies
	\begin{equation}
	P_{\sim \mu''} (\bar{Z}_{\bm{x}_\oplus} > B) \leq \sum_{i=1}^4 \epsilon_i(\eta),
	\end{equation}
	as we have claimed.
\end{proof}

Let us recall now definition of \(\mathcal{F}'\) from Eq. (\ref{eq:fprim})
\begin{equation}
\mathcal{F}' \coloneqq \left\{(y^i_0, y^i_1, x^i_A, a^i_A, x^i_{\oplus}, a^i_F)_{i=1}^n \in \{ 0,1 \}^{6n} :  \left|\left\{ i \in [n] : a^i_A = y^i_{x^i_A} \right\}\right| > \theta n,
 \left|\left\{ i \in [n] : a^i_F = y^i_{x_\oplus^i \oplus x^i_A} \right\}\right| > \theta n \right\}.
\end{equation}
In what follows for clarity we will explicitly show dependence on \(\theta\) using notation \( \mathcal{F}_{\theta}' \).

\begin{observation} \label{obs:ACCgiven}
	Let the acceptance threshold be $\theta > B/2$.	Then
	\begin{equation} \label{eq:ACCgiven}
	P(\mathcal{F}_{\theta}'\cap T(\eta)) \leq  \epsilon'_2(\eta) \coloneqq \max_{\bm{x}_\oplus \in T(\eta)} \epsilon_2(\eta, \bm{x}_\oplus),
	\end{equation}
	where \( \epsilon_2(\eta, \bm{x}_\oplus) = \sum_{i=1}^{4} \epsilon_i (\eta) \).
\end{observation}
\begin{proof}
	Let us fix $x_\oplus$. 	Then there is
	\begin{equation}
	\begin{split}
	\mathcal{F}_{\theta|\bm{x}_\oplus}'\equiv& \left\{S'(\bm{x}_\oplus): \bar{Z}^A_{\bm{x}_\oplus} \geq \theta
	\wedge \bar{Z}^F_{\bm{x}_\oplus} \geq \theta \right\}\\
	\subseteq& \left\{S'(\bm{x}_\oplus): \bar{Z}^A_{\bm{x}_\oplus} + \bar{Z}^F_{\bm{x}_\oplus} \geq 2\theta\right\}\\
	=& \left\{S'(\bm{x}_\oplus): \bar{Z}_{\bm{x}_\oplus} \geq B\right\}.
	\end{split}
	\end{equation}
	The above fact is consequences of the sequence of implications, where the last follows from $2 \theta > B$. We now invoke Eq. (\ref{eq:accgivenxor}) and note immediately a bound:
	
	\begin{equation}
	\begin{split}
	P_{\mu'}(\mathcal{F}'_{\theta} \cap T(\eta)) =& \sum_{\bm{x}_\oplus\in T(\eta)} P(\bm{x}_\oplus) P_{\mu''}(\mathcal{F}'_{\theta}|\bm{x}_\oplus)\\
	=&\sum_{\bm{x}_\oplus\in T(\eta)} P(\bm{x}_\oplus) \sum_{S'(\bm{x}_\oplus)\in \mathcal{F}'_{AF,\theta} }  P(S'(\bm{x}_\oplus))\\
	=& \sum_{\bm{x}_\oplus\in T(\eta)} P(\bm{x}_\oplus) \sum_{
		\{S'(\bm{x}_\oplus): \bar{Z}^A_{\bm{x}_\oplus}(S'(\bm{x}_\oplus))\geq \theta
		\wedge \bar{Z}^F_{\bm{x}_\oplus}(S'(\bm{x}_\oplus))\geq \theta\}} P(S'(\bm{x}_\oplus))\\
	\leq& \sum_{\bm{x}_\oplus\in T(\eta)} P(\bm{x}_\oplus) \sum_{
		\{S'(\bm{x}_\oplus): \bar{Z}_{\bm{x}_\oplus}(S'(\bm{x}_\oplus))\geq B\} } P(S'(\bm{x}_\oplus))\\ 
	\leq& \sum_{\bm{x}_\oplus \in T(\eta)}P(\bm{x}_\oplus)\epsilon_2(\eta, \bm{x}_\oplus)\\
	\leq& \max_{\bm{x}_\oplus \in T(\eta)} \sum_{i=1}^4 \epsilon_i(\eta, \bm{x}_\oplus)
	\eqqcolon  \epsilon_2(\eta).
	\end{split}
	\end{equation}
	In the pre-last inequality we have used Lemma \ref{lem:small-event} with $\epsilon_2(\eta) =\sum_{i=1}^4 \max_{\bm{x}_\oplus \in T(\eta)} \epsilon_i(\eta)$.
\end{proof}

\section{Proof of main Theorem \ref{thm:main}} \label{app:main}

Finally, after presenting all necessary definitions and lemmas, we are ready to state the complete version of Theorem \ref{thm:main} with the proof. Let
\begin{equation} \label{eq:treshold}
\begin{split}
\beta_\eta =&  P_Q\left(\frac{1}{2} + \eta\right) + \frac{M}{2} \left(\frac{1}{2}+\eta\right) + \eta\\
=& \left( \cos^2\left( \frac{\pi}{8} \right) + \frac{5 + \sqrt{3}}{8} \right) \left(\frac{1}{2} + \eta\right) +\eta\\
=&  \frac{8 \cos^2\left( \frac{\pi}{8} \right) + 5 + \sqrt{3}}{16} +    \frac{8\cos^2\left( \frac{\pi}{8} \right) + 5 + \sqrt{3} + 8}{8} \eta\\
=&  \frac{9 + 2\sqrt{2} + \sqrt{3}}{16} +    \frac{17 + 2\sqrt{2} + \sqrt{3}}{8} \eta\\	
\approx& 0.8475
+ 2.6950
\eta,
\end{split}
\end{equation} 
for some small \( \eta \) chosen in such a way that \( \beta_\eta \leq P_Q \). 
\begin{theorem}[Security of Semi-Device Independent Quantum Money] \label{thm:main}
	Let acceptance threshold \( \theta \) be larger than \(\beta_\eta n\).
	Then, under Assumptions A1-A7 (see Section \ref{ite:assumtions}), where \(k\) denotes number of Bank's branches the probability of a successful forgery \( P(\mathcal{F}_\theta) \) is exponentially small in number of banknote's qubits, and is bounded by   
	\begin{equation} \label{eq:pforge}
	P(\mathcal{F}_\theta) 
	\leq  10 k^2 \mathrm{e}^{-2\eta^2\left(\frac{1}{2} - \eta\right)n}.
	\end{equation}
\end{theorem}
\begin{proof}
	Using Corollary \ref{cor:accacc'}, Lemma \ref{lem:epstyp}, and Observation \ref{obs:ACCgiven} we obtain the following bound on \( P(\mathcal{F}) \)
	\begin{equation}
	\begin{split}
	P_{\sim \mu}(\mathcal{F}) &\stackrel{Eq. (\ref{eq:accacc'})}{=} P_{\sim \mu'}(\mathcal{F}')\\
	&\stackrel{Eq. (\ref{eq:epstyp})}{\leq} P_{\sim \mu'}(\mathcal{F}' \cap T(\eta)) + \epsilon(\eta)\\
	&\stackrel{Eq. (\ref{eq:ACCgiven})}{\leq} \max_{\bm{x}_\oplus \in T(\eta)} \epsilon_2(\eta, \bm{x}_\oplus) + \epsilon(\eta)\\
	&\leq 5 \epsilon_\eta = 10 \mathrm{e}^{-2\eta^2\left(\frac{1}{2} - \eta\right)n}.
	\end{split}
	\end{equation}
	
	Since there are many Bank's branches, collaborating Alice and Frederick can use birthday attack in order to choose two branches that have the biggest common set \(D\). For \(k\) branches we apply union bound obtaining another factor \( k^2 \) what finalizes the proof.
\end{proof}

\section{Proof of No-go for device independent quantum money} \label{app:nogo}

In this section we will provide more formal proof of Observation \ref{obs:nogo}.

\begin{proof}
	We will follow the standard proof technique used in device independent scenarios. Since all of the quantum devices are untrusted the only way for Bank's branches to verify the money is to use classical statistics of inputs and outputs from quantum black boxes. Without loss of generality we can assume four-partite scenario with two Bank's branches \( B_1 \), \( B_2 \), Alice \( A \), and Frederick \(F\) that all share some nonlocal box. We will denote device's inputs as \(I\) and outputs as \(O\) with appropriate subscript denoting the party. The general scenario can be modeled by probability distribution of the form
	\begin{equation}
	D=P(O_{B_1}, O_{B_2}, O_A, O_F | I_{B_1}, I_{B_2}, I_A, I_F).
	\end{equation}
	In verification phase Alice tries to pass verification with branch \(B_1\), while Frederick tries with branch \(B_2\). Since Bank's branches cannot communicate, they have access only to the part of outputs.	
	Therefore, the only way is to check correlations from distribution
	\begin{equation}
	D_1=P(O_{B_1}, O_A | I_{B_1}, I_A, I_{B_2}, I_F)
	\end{equation}
	for the branch \( B_1 \) and
	\begin{equation}
	D_2=P(O_{B_2}, O_F | I_{B_2}, I_F, I_{B_1}, I_A)
	\end{equation}
	for the branch \( B_2 \). We assume here that Alice and Frederic can freely talk during the verification stage.
	The money scheme, in order to be secure have to disallow both Alice and Frederick to pass verification. Furthermore, the conditions for passing verification have to be the same for all branches.  
	Since there must exist a honest quantum implementation (see Appendix \ref{aap:honest}) for Alice with distribution \( H_1 \), then Frederick could also pass verification using the same distribution \( D_2=H_1 \). To obtain such result, Adversary, controlling the source and measurement devices, can prepare joined device with distribution \(D\) simply as \( H_1 \otimes H_1 \).
	Such preparation of state and measurement will always break any device independent money scheme and cannot be detected by the Bank in any way without communication between the branches of the Bank.
\end{proof}

\section{Honest implementation} \label{aap:honest}

In this appendix we will present honest implementation based on Semi-Device Independent Quantum Key Distribution \cite{sdi}. Let, for all runs \(i\), the honest source prepare all states  \( \rho^i_{y^i_0,y^i_1} \) in the following way
\begin{equation}
\rho_{00} \coloneqq \ket{0}\bra{0},\quad \rho_{01} \coloneqq \ket{-}\bra{-},\quad \rho_{10} \coloneqq \ket{+}\bra{+},\quad \rho_{11} \coloneqq \ket{1}\bra{1}
\end{equation}
where \( \ket{\pm} \coloneqq (\ket{0} \pm \ket{1}) / \sqrt{2} \). Let us also choose an appropriate measurement \( M_{x^i} \), depending on the branch's question \(x^i\),
\begin{equation}
M_0 \coloneqq \frac{\sigma_z + \sigma_x}{\sqrt{2}} ,\quad M_1 \coloneqq \frac{\sigma_z - \sigma_x}{\sqrt{2}},
\end{equation}
where \( \sigma_x \) and \( \sigma_z \) are Pauli matrices 
\begin{equation}
\sigma_x \coloneqq \begin{bmatrix}
0 & 1\\
1 & 0
\end{bmatrix},\quad
\sigma_z \coloneqq \begin{bmatrix}
1 & 0\\
0 & -1 
\end{bmatrix}.
\end{equation}
It turns out, as shown in \cite{sdi}, that, using these states and measurements, the optimal guessing probability equals \( P_Q = \cos^2(\pi/8) \).
\begin{remark}[Connection with Wiesne'sr money scheme]
	In the honest implementation of our scheme we use the same states as in the original Wiesner scheme. On the other hand the measurement settings have to be different since, from \cite{sdi}, we know that Wiesner's can not be used in the semi-device independent approach.
\end{remark}

\section{Required number of qubits} \label{app:mem}

This appendix establishes the relation between the number of qubits and the upper bounds on the probability of forgery.
Since from the Eq. (\ref{eq:treshold}) we know that \( \beta_\eta \leq P_Q \) we can calculate that the maximal allowed value of \( \eta \)  equals
\begin{equation}
	\eta_{\text{max}} \coloneqq \frac{-1 +2\sqrt{2} - \sqrt{3}}{34 + 4\sqrt{2} + 2\sqrt{3}} \approx 0.0022.
\end{equation}
When we put that value into Eq. (\ref{eq:pforge}), for the trivial case of the single Bank without any additional branches we obtain that 
\begin{equation}
	P(\mathcal{F}) \leq  10 \mathrm{e}^{-2\eta_{\text{max}}^2\left(\frac{1}{2} - \eta_{\text{max}}\right)n}.
\end{equation}
It is easy to calculate numerically that this bound becomes trivial when the number of qubits \(n\) is smaller than 463018. Furthermore when we demand that the probability of forgery is smaller than some security parameter and we want to assume more realistic scenario the number of required qubits grows significantly. 

Although one cannot expect that such a large number of qubits will be available in quantum memories soon, let us emphasize that the bounds used in the proof of Theorem \ref{thm:main} are not tight, and there is some room for improvement.
What is more important, we expect that using a more complex random access codes i.e., ones with more inputs and outputs can lead to significant decrease of the number of required qubits as it is discussed in Section \ref{sec:discusion}.

\twocolumngrid

\bibliography{SDIQMbibliography}

\end{document}